%% file: Corr.tex
\documentclass{llncs}
\usepackage{microtype}

\usepackage{amsmath}
\usepackage{amssymb}
\usepackage{mathtools}
\usepackage{verbatim} 
\usepackage{stmaryrd}
\usepackage{url}
\allowdisplaybreaks
\usepackage{basic}
\usepackage{ambients}
\usepackage{paralist}
\usepackage{graphicx}
\usepackage{proof}
\usepackage{hyperref}

\usepackage{enumitem}

\usepackage{epsf}
\usepackage{graphics} 
\usepackage{epsfig} 

\newcommand{\transStep}[1]{\xrightarrow{\, {#1} \, }} 
\newcommand{\TransStep}[1]{\xRightarrow{\, {#1} \, }} 
 
\newcommand{\ntransStep}[1]{\mathrel{{\transStep{#1}}\makebox[0em][r]{$\not$\hspace{2ex}}}{\!}}
\newcommand{\nTransStep}[1]{\mathrel{{\TransStep{#1}}\makebox[0em][r]{$\not$\hspace{2ex}}}{\!}}

\renewcommand{\timeout}[2]{\lfloor #1 \rfloor #2  }
\renewcommand{\tick}{\mathsf{tick}}

\newcommand{\der}{\mathit{der}}

\newif\ifdraft\drafttrue

\input{macros_tip}
\renewcommand{\operatorname}[1]{\mathit{#1}}

\newcommand{\states}{\mathcal{T}}

\newcommand{\bisimtoll}[2]{\mathrel{\approx^{#1}_{#2}}}
\newcommand{\transSimone}[1][]{\xrightarrow{\, {#1} \, }}
\newcommand{\TransSimone}[1][]{\xRightarrow{\, {#1} \, }}
\newcommand{\transSim}[1]{\xrightarrow{\, {#1} \, }}

\newcommand{\CPS}{CPS}
\newcommand\restrict[1]{\raise-.5ex\hbox{\ensuremath|}_{#1}}

\newtheorem{notation}{Notation}

\usepackage{marvosym}

\newcommand{\distr}[1]{{\mathcal D}(#1)}

\DeclareMathOperator{\Kantorovich}{\mathbf{K}} 
\newcommand{\size}[1]{\mid\!\!{#1}\!\!\mid}

\makeatletter
\renewcommand{\xRightarrow}[2][]{\ext@arrow 0359\Rightarrowfill@{#1}{#2}}
\newcommand{\xRightarrowBis}[2][]{\ext@arrow 0359\Rightarrowfill@{#1}{#2}}
\makeatother

\newcommand{\distrib}{\mathcal D}
\newcommand{\dummyN}{\mathsf{Dead}}

\newcommand{\metric}{\ensuremath{\mathbf{m}}}

\DeclareMathOperator{\zeroF}{{\bf 0}}
\DeclareMathOperator{\oneF}{{\bf 1}}
\DeclareMathOperator{\Bisimulation}{\mathbf{B}} 
\DeclareMathOperator{\Simulation}{\mathbf{S}} 
\newcommand{\dirac}[1]{\overline{#1}}
\renewcommand{\rcva}[2]{#1?(#2)}

\definecolor{darkred}{RGB}{128,0,0}
\definecolor{darkgreen}{RGB}{0,128,0}
\definecolor{lightgreen}{RGB}{224,255,224}

\newcommand{\statefun}{\xi_{\mathrm{x}}}
\newcommand{\actuatorfun}{\xi_{\mathrm{a}}}

\newcommand{\evolmap}{\mathit{evol}}

\newcommand{\measmap}{\xi_{\mathrm{m}}}
\newcommand{\invariantfun}{\mathit{inv}}

\begin{document}

\title{Towards a formal notion of impact metric for cyber-physical attacks (full version)\thanks{An extended abstract will appear in the Proc.\ of the \emph{14th International Conference on integrated Formal Methods} (iFM 2018), 5th-7th September 2018, Maynooth University, Ireland, and published in a volume of  \emph{Lecture Notes in Computer Science}.}}

\titlerunning{Towards a formal notion of impact metric for cyber-physical attacks}

\author{Ruggero Lanotte\inst{1} \and Massimo Merro\inst{2} \and
Simone Tini\inst{1}}

\authorrunning{R. Lanotte et al.}

\institute{Dipartimento di Scienza e Alta Tecnologia, Universit\`a dell'Insubria, Como, Italy \\ \email{\{ruggero.lanotte,simone.tini\}@uninsubria.it}
\and Dipartimento di Informatica, Universit\`a degli Studi di Verona, Verona,  Italy \email{massimo.merro@univr.it}}

\maketitle

\begin{abstract}
Industrial facilities and critical infrastructures are transforming into ``smart" environments that dynamically adapt to external events. The result is an ecosystem of heterogeneous physical and cyber components integrated in cyber-physical systems which are more and more exposed to \emph{cyber-physical attacks}, 
\emph{i.e.},  security breaches in cyberspace that adversely affect the physical processes at the core of the systems. 

We provide a formal \emph{compositional metric} to estimate the \emph{impact} of cyber-physical attacks targeting sensor devices of \emph{I{o}T systems} formalised in a simple extension of Hennessy and Regan's \emph{Timed Process Language}.
 Our \emph{impact metric} relies on a discrete-time generalisation of Desharnais et al.'s \emph{weak bisimulation metric} for \nolinebreak con\-cur\-rent systems. We show the adequacy of our definition on two different attacks on a simple surveillance system. 
\end{abstract}


\section{Introduction}
The \emph{Internet of Things} (IoT) is heavily affecting our daily lives in many domains, ranging from tiny wearable devices to large industrial systems with thousands of heterogeneous cyber and physical components that interact with each other.

\emph{Cyber-Physical Systems} (\CPS{s}) are integrations of networking and distributed computing systems with physical processes, where feedback loops allow the latter to affect the computations of the former and vice versa. 
Historically, \CPS{s} relied on proprietary technologies and were implemented as stand-alone networks in physically protected locations. However,  the growing connectivity and integration of these systems  has triggered a  dramatic increase in the number of \emph{cyber-physical attacks}~\cite{CPS-Book2015}, \emph{i.e.},  security breaches in cyberspace that adversely affect the physical processes, 
 \emph{e.g.}, manipulating \emph{sensor readings} and, in general, influencing physical processes to bring the system into a state desired \nolinebreak by \nolinebreak  the \nolinebreak attacker.

Cyber-physical attacks are complex and challenging as they usually cross the boundary between cyberspace and the physical world, possibly more than once~\cite{GGIKLW2015}. Some notorious examples are: (i) the \emph{Stuxnet} worm, which reprogrammed PLCs of nuclear centrifuges in Iran~\cite{stuxnet},  (ii) the attack on a sewage treatment facility in Queensland, Australia,   which manipulated the SCADA system to release raw sewage into local rivers~\cite{SlMi2007},  or the (iii) the recent \emph{BlackEnergy} cyber-attack on the Ukrainian power grid, again  compromising the SCADA system~\cite{ICS15}. 

The points in common of these systems   is that they are all safety critical and failures may cause catastrophic consequences. Thus, the concern for consequences at the physical level puts \emph{\CPS{} security} apart from standard \emph{IT security}. 

\emph{Timing} is particularly relevant in \CPS{} security because the physical state of a system changes continuously over time and, as the system evolves in time, some states might be more vulnerable to attacks than others~\cite{KrCa2013}. For example, an attack launched when the target state variable reaches a local maximum (or minimum) may have a great impact on
the whole system behaviour~\cite{BestTime2014}. 
Also the \emph{duration of the attack} is an
important parameter to be taken into consideration in order to achieve a
successful attack. For example, it may take minutes for a chemical reactor
to rupture~\cite{chemical-reactor}, hours to heat a tank of water or burn
out a motor, and days to destroy centrifuges~\cite{stuxnet}.

Actually, the estimation of the \emph{impact} of cyber-physical attacks on the target system  is crucial when protecting \CPS{s}~\cite{GeKiHa2015}. 
For instance, in industrial CPSs, before taking any countermeasure against an attack, engineers 
 first try to estimate the impact of the attack on the system functioning (e.g., performance and security) and weight it against the cost of stopping the plant. If this cost is higher than the damage caused by the attack (as is sometimes the case), then engineers might actually decide to let the system continue its activities even under attack. Thus, once an attack is detected, \emph{impact metrics} are necessary to quantify the perturbation introduced in the physical behaviour of the system under attack. 

The \emph{goal} of this paper is to lay theoretical foundations to provide formal instruments to precisely define the notion of impact of cyber-physical attack targeting physical devices, such as \emph{sensor devices} of IoT systems. For that we rely on a timed generalisation of \emph{bisimulation metrics}~\cite{DJGP02,DGJP04,BW05} 
 to compare the behaviour of two systems up to a given tolerance, for time-bounded executions.

\emph{Weak bisimulation metric}~\cite{DJGP02} allows  us to compare two systems $M$ and $N$, writing $M \simeq_{p} N$, if the weak bisimilarity holds with a \emph{distance} or \emph{tolerance} $p \in [0,1]$, \emph{i.e.}, if $M$ and $N$ exhibit a different behaviour with probability $p$, and the same behaviour with probability $1-p$. 
A useful generalisation  is  the \emph{$n$-bisimulation metric}~\cite{vB12} that takes into account bounded computations. Intuitively, the distance $p$ 
is ensured only for the first $n$ computational steps, for some $n \in \mathbb{N}$.
However, 
in timed systems it is desirable  to focus on the passage of time rather than the number of computational steps. This would allow us to deal with situations where it is not necessary (or it simply does not make sense) to compare two systems ``ad infinitum'' but only for a limited amount of time.

\subsubsection*{Contribution.}
In this paper, we first introduce a general notion of \emph{timed bisimulation metric} for concurrent probabilistic systems equipped with a discrete notion of time. Intuitively, this kind of  metric allows us to derive a \emph{timed weak bisimulation with tolerance}, denoted with $\approx_p^k$, for $k\in \mathbb{N}^+ \cup \{ \infty \}$  and $p \in [0,1]$, to express that the tolerance $p$ between  two timed systems  is ensured only for the first $k$ time instants ($\tick$-actions). Then, we use our timed bisimulation metric to set up 
a formal \emph{compositional} theory  to study and  measure the \emph{impact} 
of cyber-physical attacks on IoT systems specified in a simple probabilistic timed process calculus which 
extends Hennessy and Regan's \emph{Timed Process Language} (TPL)~\cite{HR95}. 
IoT systems in our calculus are modelled by specifying:
\begin{inparaenum}[(i)]
\item a \emph{physical environment\/},  containing informations on the 
physical state variables and the sensor measurements, and 
\item a \emph{logics} that governs both accesses to sensors and channel-based communications with other cyber components. 
\end{inparaenum}

We focus on \emph{attacks on sensors}
that may eavesdrop and possibly modify the sensor measurements provided to the controllers of  sensors, affecting both the \emph{integrity} and the \emph{availability} of the system under attack. 

In order to make security assessments of our IoT systems, we   adapt 
a well-know approach called \emph{Generalized Non Deducibility on Composition} (GNDC)~\cite{FM99} to compare the behaviour of an IoT system $M$  with the behaviour of the same system under attack, written  $M \parallel A$, for some arbitrary cyber-physical attack $A$. 
This comparison makes use of our timed bisimulation metric  to evaluate not only the \emph{tolerance} and the \emph{vulnerability} of a  system $M$ with respect to a certain attack $A$, but also  the \emph{impact} of a successful attack  in terms of the deviation introduced  in the behaviour of the target system. In particular, we say that a system $M$ \emph{tolerates an attack} $A$ if $M \parallel A \approx^{\infty}_0 M $, \emph{i.e.}, the presence of $A$ does not affect the behaviour of $M$; whereas $M$ is said to be \emph{vulnerable} to $A$ in the time interval $m..n$ with impact $p$ if $m..n$ is the smallest interval  such that 
$M \parallel A \approx^{m-1}_0 M $ and $M  \parallel A \approx^{k}_p M $, for any $k \geq n$, \emph{i.e.}, if the perturbation introduced by the attack $A$ becomes observable in the $m$-th time slot and yields the maximum \emph{impact} $p$ in the $n$-th time slot. 
In the concluding discussion we will show that the \emph{temporal vulnerability window} $m..n$ provides several informations about the corresponding attack, such as \emph{stealthiness} capability, duration of the \emph{physical effects} of the attack, and  consequent room for possible run-time \emph{countermeasures}. 

As a case study, we use our timed bisimulation metric  to measure the impact of two different attacks injecting \emph{false positives} and \emph{false negative}, respectively, into a simple surveillance system expressed in our process calculus.  

\subsubsection*{Outline.} Section~\ref{sectionMainDefinitions} formalises our timed bisimulation metrics in a general setting. Section~\ref{sec:impact} provides a simple calculus of IoT systems. 
Section~\ref{sec:cyber-physical-attackers} defines cyber-physical attacks together with the notions of tolerance and vulnerability 
\emph{w.r.t.}\ an attack.  In Section~\ref{sec:case} we use our metrics to evaluate the impact of two attacks on a simple surveillance system. Section~\ref{sec:conclusions} draws conclusions and discusses related and future \nolinebreak work. In this extended abstract proofs are omitted, full details of the proofs can be found in the Appendix.

\section{Timed Bisimulation Metrics}
\label{sectionMainDefinitions}

In this section, we introduce \emph{timed bisimulation metrics} as a general instrument to derive a notion of timed and approximate weak bisimulation between probabilistic systems equipped with a  discrete notion of time. 
In  \autoref{sec:PTS},  we recall the semantic model of \emph{nondeterministic probabilistic labelled transition systems}; in \autoref{sec:new},  we present our metric semantics.

\subsection{Nondeterministic Probabilistic Labelled Transition Systems}
\label{sec:PTS} 
Nondeterministic probabilistic labelled transition systems (pLTS) \cite{S95} combine classic LTSs \cite{K76} and discrete-time Markov chains \cite{HJ94,Ste94}  to model, at the same time, reactive behaviour, nondeterminism and probability.
We first provide the mathematical machinery required to define a pLTS. 

The state space in a pLTS is given by a set $\states$, whose elements are called $\emph{processes}$, or $\emph{terms}$. We use $t,t',..$ to range over $\states$. A (discrete) \emph{probability sub-distribution} over $\states$ is a mapping $\Delta \colon \states \to [0,1]$, with $\sum_{t \in \states}\Delta(t) \in (0 , 1]$. We denote $\sum_{t \in \states}\Delta(t)$ by $\size{\Delta}$, and we say that $\Delta$ is a \emph{probability distribution} if $\size{\Delta}=1$. The \emph{support} of $\Delta$ is given by $\lceil \Delta \rceil = \{ t \in \states : \Delta(t) > 0 \}$. The set of all sub-distributions (resp. distributions) over $\states$ with finite support will be denoted with ${\mathcal D}_{\mathrm{sub}}(\states)$ (resp. ${\mathcal D}(\states)$). We use $\Delta$, $\Theta$, $\Phi$ to range over ${\mathcal D}_{\mathrm{sub}}(\states)$ and ${\mathcal D}(\states)$.
\begin{definition}[pLTS~\cite{S95}]
\label{def:pLTS}
A \emph{pLTS} is a triple $(\states,\Act,\transStep{})$, where: 
\begin{inparaenum}[(i)]
\item $\states$ is a countable set of \emph{terms}, 
\item $\Act$ is a countable set of \emph{actions}, and 
\item $\transStep{} \, \subseteq {\states \times \Act \times \distr{\states}}$ is a \emph{transition relation}. 
\end{inparaenum}
\end{definition}
In Definition~\ref{def:pLTS}, we assume the presence of a special deadlocked term $\dummyN \in \states$. Furthermore, we assume that the set of actions $\Act$ contains at least two actions: $\tau$ and $\tick$. 
The former to model internal computations that cannot be externally observed, 
 while the latter denotes the passage of one time unit in a setting with a discrete notion of time~\cite{HR95}. In particular, $\tick$ is the only \emph{timed action} in $\Act$.

We write $t \transStep{\alpha} \Delta$ for $(t,\alpha,\Delta)\!\in \, \transStep{}$, 
$t \transStep{\alpha}$ if there is a distribution $\Delta \in \distr{\states}$ with $t \transStep{\alpha} \Delta$, and $t \ntransStep{\alpha}$ otherwise. 
Let $\der(t,\alpha) =\{\Delta\in\distr{\states} \mid t \transStep{\alpha}\Delta\}$ denote the set of the derivatives (i.e.\ distributions) reachable from term $t$ through action $\alpha$.
We say that a pLTS  is \emph{image-finite} \cite{HPSWZ11} if $\der(t,\alpha)$ is finite for all $t \in \states$ and $\alpha \in \Act$.
In this paper, we will always work with  image-finite pLTSs. \\[4pt] 
\noindent 
\emph{Weak transitions.} 
As we are interested in developing a \emph{weak} bisimulation metric, we need a definition of weak transition 
 which abstracts away from $\tau$-actions.
\label{sec:weak}
In a probabilistic setting, the definition of weak transition is somewhat complicated by the fact that (strong) transitions take terms to distributions;
consequently if we are to use weak transitions  then we need to generalise transitions, so that they take (sub-)distributions to (sub-)distributions.

To this end,  we need some extra notation on distributions.
For a term $t \in \states$, the \emph{point (Dirac) distribution at $t$\/}, denoted $\dirac{t}$, is defined by $\dirac{t}(t) = 1$ and $\dirac{t}(t') = 0$ for all $t' \neq t$.
Then, the convex combination $\sum_{i \in I}p_i \cdot  \Delta_i$ of a family $\{\Delta_i\}_{i \in I}$ of (sub-)distributions, with $I$ a finite set of indexes, $p_i \in (0,1]$ and $\sum_{i \in I}p_i \le 1$, is the (sub-)distribution defined by
$(\sum_{i \in I}p_i \cdot  \Delta_i)(t)  \deff \sum_{i\in I}p_i \cdot \Delta_i(t)$
for all $t \in \states$.
We write $\sum_{i \in I}p_i \cdot  \Delta_i$ as $p_1 \cdot \Delta_1 + \ldots + p_n \cdot \Delta_n$ when $I = \{ 1, \ldots , n \}$.

Along the lines of~\cite{Dengetal2008}, we write $t \transStep{\hat{\tau}} \Delta$, for some term $t$ and some distribution $\Delta$,  if either $t \transStep{\tau} \Delta$ or $\Delta = \dirac{t}$. 
Then, for $\alpha \neq \tau$, we write $t \transStep{\hat{\alpha}} \Delta$ if $t \transStep{\alpha} \Delta$.
Relation $\transStep{\hat{\alpha}}$ is extended to model transitions from sub-distributions to sub-distributions.
For a sub-distribution $\Delta =\sum_{i \in I}p_i \cdot \dirac{t_i}$, we write $\Delta \transStep{\hat{\alpha}} \Theta$ if there is a non-empty set of indexes $J\subseteq I$ such that: 
\begin{inparaenum}[(i)]
\item
$t_j \transStep{\hat{\alpha}} \Theta_j$ for all $j \in J$, 
\item
$t_i \ntransStep{\hat{\alpha}}$, for all $i \in I \setminus J$,  and 
\item $\Theta = \sum_{j \in J}p_j  \cdot \Theta_j$.
\end{inparaenum}
Note that if $\alpha \neq \tau$ then this definition admits that only some terms in the support of $\Delta$ make the $\transStep{\hat{\alpha}}$ transition.
Then, we define the \emph{weak transition relation} $\TransStep{\hat{\tau}}$ as the transitive and reflexive closure of $\transStep{\hat{\tau}}$, \emph{i.e.}, 
$\TransStep{\hat{\tau}} \, = (\transStep{\hat{\tau}})^{\ast}$, while for $\alpha \neq \tau$ we let $\TransStep{\hat{\alpha}}$ denote $\TransStep{\hat{\tau}} \transStep{\hat{\alpha}} \TransStep{\hat{\tau}}$.

\subsection{Timed Weak Bisimulation with Tolerance}
\label{sec:new}
In this section, we define a family of relations 
$\bisimtoll{k}{p}$ over $\states$, with $p \in [0,1]$ and $k \in \mathbb{N}^+ \cup \{ \infty \}$, where, intuitively, $t \bisimtoll{k}{p} t'$ means that \emph{$t$ and $t'$ can weakly bisimulate each other with a tolerance $p$ accumulated in $k$ timed steps}.
This is done by introducing a family of \emph{pseudometrics} $\metric^k \colon \states \times \states \to [0,1]$ and defining $t \bisimtoll{k}{p} t'$ iff $\metric^{k}(t,t') = p$. The pseudometrics $\metric^k$ will have the following properties for any $t,t' \in \states$:\!
\begin{inparaenum}[(i)]
\item 
$\metric^{k_1}(t,t') \le \metric^{k_2}(t,t')$ whenever $k_1 < k_2$ (tolerance monotonicity);
\item
$\metric^{\infty}(t,t') = p$ iff $p$ is the distance between $t$ and $t'$ as given by the weak bisimilarity metric in \cite{DJGP02} in an untimed setting; 
\item
$\metric^{\infty}(t,t') = 0$ iff $t$ and $t'$ are related by the standard weak probabilistic bisimilarity~\cite{ALS00}.
\end{inparaenum}

Let us recall the standard definition of pseudometric.
\begin{definition}[Pseudometric]
\label{def:pseudoquasimetric}
A function $d \colon \states \times \states  \to [0,1]$ is a \emph{1-bounded pseudometric} over $\states$ if
\begin{itemize}	
	\item $d(t,t)= 0$ for all $t \in \states$,
         \item $d(t,t') = d(t',t)$ for all $t,t' \in \states$ (symmetry),
	\item $d(t,t') \le d(t,t'') + d(t'',t')$ for all $t,t',t''\in \states$ (triangle inequality).
\end{itemize}
\end{definition}

In order to define the family of functions 
$\metric^{k}$, we define an auxiliary family of functions  $\metric^{k,h} \colon \states \times \states  \to [0,1]$, with $k,h \in \mathbb{N}$, 
quantifying the tolerance of the weak bisimulation after a sequence of computation steps such that:
\begin{inparaenum}[(i)] 
\item the sequence contains exactly $k$ $\tick$-actions,
\item the sequence terminates with a $\tick$-action,
\item any term performs exactly $h$ untimed actions before the first $\tick$-action, 
\item between any $i$-th and $(i{+}1)$-th $\tick$-action,  with $1\le i < k$,  
there are an arbitrary number of untimed actions. 
\end{inparaenum}

The definition of $\metric^{k,h}$ 
relies on a \emph{timed and quantitative} version of the classic bisimulation game: 
The tolerance between $t$ and $t'$ as given by $\metric^{k,h}(t,t')$ can be below
a threshold $\epsilon \in [0,1]$ only if each transition $t \transSimone[\alpha] \Delta$ is mimicked by a weak transition $t' \TransSimone[\hat{\alpha}] \Theta$ such that the bisimulation tolerance between $\Delta$ and $\Theta$ is, in turn, below $\epsilon$.
This requires to lift pseudometrics over $\states$ to pseudometrics over (sub-)distributions in ${\mathcal D}_{\mathrm{sub}}(\states)$.
To this end, 
we adopt the notions of \emph{matching}~\cite{Vil08} (also called coupling) and \emph{Kantorovich lifting\/}~\cite{Den09}.
\begin{definition}[Matching]
\label{def_matching}
A \emph{matching} for a pair of distributions 
$(\Delta,\Theta) \in {\mathcal D}(\states) \times {\mathcal D}(\states)$ is a distribution $\omega$ in the state product space ${\mathcal D}(\states \times \states)$ 
such that: 
\begin{itemize}
\item 
$\sum_{t' \in \states} \omega(t,t')=\Delta(t)$, for all $t \in \states$, and 
\item 
 $\sum_{t \in \states} \omega(t,t')=\Theta(t')$, for all $t' \in \states$. 
\end{itemize}
We write $\Omega(\Delta ,\Theta)$ to denote the set of all matchings for $(\Delta,\Theta)$.
\end{definition}
A matching for $(\Delta,\Theta)$ may be understood as a transportation schedule for the shipment of probability mass from $\Delta$ to $\Theta$ \cite{Vil08}.

\begin{definition}[Kantorovich lifting] \label{def:KantorovichLifting}
\label{def:Kantorovich}
Assume a pseudometric $d\colon \states \times \states  \to [0,1]$.
 The \emph{Kantorovich lifting} of $d$ is the function
$\Kantorovich(d) \colon \distr{\states} \times \distr{\states} \to [0,1]$ defined for distributions $\Delta$ and $\Theta$ 
as: 
\begin{center}
\(
\Kantorovich(d)(\Delta,\Theta) \deff  \min_{\omega \in \Omega(\Delta,\Theta)} \sum_{s,t \in \states}\omega(s,t) \cdot d(s,t).
\)
\end{center}
\end{definition}
Note that since we are considering only distributions with finite support, the minimum over the set of matchings $\Omega(\Delta,\Theta)$ used in \autoref{def:Kantorovich} is well defined.

Pseudometrics $\metric^{k,h}$  are  inductively defined on $k$ and $h$ by means of suitable \emph{functionals} over the complete lattice  ${([0,1]^{\states \times \states},\sqsubseteq)}$ of functions of type $\states \times \states \to [0,1]$, 
ordered by $d_1 \sqsubseteq d_2$ iff $d_1(t, t') \le d_2(t,t')$ for all $t,t' \in \states$. Notice that in this lattice, for each set $D \subseteq [0,1]^{ \states\times  \states}$, the supremum and infimum are defined as $\sup(D)(t,t') = \sup_{d \in D}d(t,t')$ and $\inf(D)(t,t') = \inf_{d \in D}d(t,t')$, for all $t,t' \in \states$. The infimum of the lattice is the constant function zero, denoted by $\zeroF$, and the supremum is the constant function one, denoted by $\oneF$.
\begin{definition}[Functionals for $\metric^{k,h}$] 
\label{def:metric_sim_functional}
The functionals $\Bisimulation, \Bisimulation_{\tick} \colon [0,1]^{\states\times \states } \to [0,1]^{ \states  \times \states}$
are defined for any function $d \in [0,1]^{\states\times \states }$
and terms $t,t' \in \states$ as: \\[2pt]
\begin{math}
\begin{array}{rccl}
\Bisimulation(d)(t,t') 
&=&
\displaystyle \max\{ & d(t,t')  , \\
 &&&
\displaystyle  \sup_{\alpha \in \Act{\setminus} \{ \tick  \}} \; 
 \max_{t \transStep{\alpha}\Delta} \; \inf_{t' \TransStep{\hat{\alpha}} \Theta} \Kantorovich(d)\big(\Delta,\Theta + (1-\size{\Theta})\dirac{\dummyN} \big),\\
& & &
\displaystyle  \sup_{\alpha \in \Act {\setminus} \{ \tick \} } \; 
 \max_{t' \transStep{\alpha}\Theta} \; \inf_{t \TransStep{\hat{\alpha}} \Delta} \Kantorovich(d)\big(\Delta + (1-\size{\Delta}) \dirac{\dummyN}, \Theta \big) \:\}
\\[1.5 ex] 
\Bisimulation_{\tick}(d)(t,t') 
&=&
\displaystyle  \max\{ & d(t,t')  , \\
 & &&
\displaystyle 
 \max_{t \transStep{\tick}\Delta}\;  \inf_{t' \TransStep{\widehat{\tick}} \Theta} \Kantorovich(d)\big(\Delta,\Theta + (1-\size{\Theta})\dirac{\dummyN} \big),\\
& & &
\displaystyle   \max_{t' \transStep{\tick}\Theta} \; \inf_{t \TransStep{\widehat{\tick}} \Delta} \Kantorovich(d)\big(\Delta + (1-\size{\Delta}) \dirac{\dummyN}, \Theta \big) \; \}
\end{array}
\end{math} \\
where $\inf \emptyset = 1$ and $\max \emptyset = 0$.
\end{definition}

Notice that all $\max$ in \autoref{def:metric_sim_functional} are well defined since the pLTS is image-finite.
Notice also that any strong transitions from $t$ to a distribution $\Delta$ is mimicked by a weak transition from $t'$, which, in general, takes to a sub-distribution $\Theta$.
Thus,  process $t'$ may not simulate $t$ with probability $1{-}\size{\Theta}$.

\begin{definition}[Timed weak bisimilarity metrics]
\label{twbm}
The family of the  \emph{timed weak bisimilarity metrics} $\metric^k \colon (\states \times \states) \to [0,1]$  is  defined for all $k \in \mathbb{N}$ by
\(
\metric^{k} = 
\begin{cases}
\zeroF & \text{ if } k = 0 \\
 \sup_{h \in \mathbb{N}}\metric^{k,h} 
& \text{ if } k > 0 
\end{cases}
\)

\noindent 
while 
the functions $\metric^{k,h }\colon(\states \times \states) \! \to \! [0,1]$ are defined for all $k \in \mathbb{N}^+$ and \nolinebreak  $h \in \mathbb{N}$ \nolinebreak  by\\
\(
\metric^{k,h} = \begin{cases} 
\displaystyle\Bisimulation_{\mathit{\tick}}(\metric^{k-1}) & \text{ if } h = 0 \\
\Bisimulation(\metric^{k,h-1})  & \text{ if } h > 0.
\end{cases}
\)

\noindent
Then, we define $\metric^{\infty} \colon (\states \times \states) \to [0,1]$ as $\metric^{\infty} = \sup_{k \in \mathbb{N}}\metric^{k}$.
\end{definition}
Note that any $\metric^{k,h}$ is obtained from  $\metric^{k-1}$ by one application of the functional $\Bisimulation_{\tick}$, in order to take into account the distance between terms introduced by 
the $k$-th $\tick$-action, and $h$ applications of the functional $\Bisimulation$, in order to lift such a distance to terms that take $h$ untimed actions to be able to perform a $\tick$-action. By taking $\sup_{h \in \mathbb{N}} \metric^{k,h}$ we consider an arbitrary number of untimed steps.

The pseudometric property of $\metric^k$ is necessary to conclude that 
the tolerance between terms as given by $\metric^k$ is a reasonable notion of behavioural distance.

\begin{theorem}
\label{q_and_m_are_metrics}
For any  $k \ge 1$,
$\metric^{k}$ is a 1-bounded pseudometric.
\end{theorem}

Finally, everything is in place to define our timed weak bisimilarity 
$\bisimtoll{k}{p}$ with tolerance $p \in [0, 1]$ accumulated after $k$
time units, for $k \in \mathbb{N} \cup \{\infty\}$. 
\begin{definition}[Timed weak bisimilarity with tolerance]
\label{def:distance-n}
Let $t, t' \in \states$, $k \in \mathbb{N}$ and $p \in [0,1]$. 
We say that \emph{$t$ and $t'$ are 
weakly bisimilar with a tolerance $p$, which accumulates in $k$ timed actions}, 
written $t \bisimtoll{k}{p} t'$, if and only if $\metric^{k}(t,t') = p$.
Then, we write
$t \bisimtoll{\infty}{p} t'$  if and only if $\metric^{\infty}(t,t') = p$.
\end{definition}

Since the Kantorovich lifting $\Kantorovich$ is monotone \cite{Pan09}, it follows that both functionals $\Bisimulation$ and $\Bisimulation_{\tick}$ are monotone.
This implies that, for any $k\geq 1$,  $(\metric^{k,h})_{h \ge 0}$ is a non-decreasing chain and,
analogously, also $(\metric^k)_{k \ge 0}$ is a non-decreasing chain, thus giving the following expected result saying that the distance between terms grows when we consider a higher number of $\tick$ computation steps.
\begin{proposition}[Tolerance monotonicity]
\label{prop:tol-monotonicity}
For all terms $t,t' \in \states$ and $k_1,k_2 \in \mathbb{N}^+$ with $k_1 < k_2$,  $t \bisimtoll{k_1}{p_1} t'$ and $t \bisimtoll{k_2}{p_2} t'$ entail $p_1 \le p_2$.
\end{proposition}

We conclude this section by comparing our  behavioural distance with the  behavioural relations known in the literature.

We recall that in \cite{DJGP02} a family of relations $\simeq_p$ for \emph{untimed}
process calculi  
are defined such that  $t \simeq_p t'$ if and only if $t$ and $t'$ weakly bisimulate each other with tolerance $p$.
Of course, one can apply these relations also to timed process calculi, the effect being that timed actions are treated in exactly the same manner as untimed actions. 
The following result compares the behavioural metrics proposed in the present paper with those of \cite{DJGP02},  and with the classical notions of probabilistic weak bisimilarity~\cite{ALS00} denoted $\approx$.

\begin{proposition}
\label{prop_simulazione}
Let $t,t' \in \states$ and $p \in [0,1]$. Then, 
\begin{itemize}
\item
\label{prop_simulazione_uno}
$t \bisimtoll{\infty}{p} t'$ iff $t \simeq_p t'$ 
\item
\label{prop_simulazione_due}
$t \bisimtoll{\infty}{0} t'$ iff $t \approx t'$. 
\end{itemize}
\end{proposition}


\section{A Simple Probabilistic Timed Calculus for IoT Systems}
\label{sec:impact}

In this section, we propose a simple  extension of Hennessy and Regan's \emph{timed process algebra} TPL~\cite{HR95}  to express \emph{IoT systems} and 
\emph{cyber-physical attacks\/}. The goal is to show that timed weak bisimilarity with tolerance is a suitable notion to estimate the impact of cyber-physical attacks on IoT systems.

Let us start with some preliminary notations. 
\begin{notation}
We use 
$x, x_k$ for \emph{state variables},  
$c,c_k,$ for \emph{communication channels}, 
$z_,z_k$ for \emph{communication variables}, 
 $s,s_k$ for \emph{sensors devices}, while 
$o$ ranges over both channels and sensors. \emph{Values}, ranged over by $v,v'$, belong to a \emph{finite} set of admissible values $\mathcal V$. 
We use $u, u_k$ for both values and communication variables. Given a generic set of names $\cal N $, we write $\mathcal{V}^{\cal N} $ to denote the set of functions $\mathcal N \rightarrow \mathcal{V} $ assigning a value to each name in $\mathcal N$. For  $m \in \mathbb{N}$ and $n \in \mathbb{N} \cup \{ \infty \}$, we write $m..n$ to denote an \emph{integer interval}. As we will adopt a discrete notion of time, we will use integer intervals to denote \emph{time intervals}.
\end{notation}

\emph{State variables} are associated to physical properties like 
\emph{temperature}, \emph{pressure}, etc. \emph{Sensor names}
are metavariables for sensor devices, such as 
\emph{thermometers} and \emph{barometers}. Please, notice that in cyber-physical systems, state variables cannot be directly  accessed but they can only be  
tested via one or more sensors.

\begin{definition}[IoT system]
\label{def:SmartSys}
Let $\mathcal{X}$ be a set of state variables and $\mathcal S$ be a set of sensors. Let $\mathit{range}: \mathcal X \rightarrow 2^{\mathcal{V}}$ be a total function 
returning the range of admissible values for any state variable $x \in \mathcal X$.  An \emph{IoT system}  consists of two components: 
\begin{itemize}[noitemsep]
\item a \emph{physical environment} $\xi = \langle \statefun{} , \measmap{} \rangle$
where: 
\begin{itemize}
\item  
$\statefun{} \in  \mathcal{V}^{\mathcal X}$ is the \emph{physical state}  of the system that associates a value to each state variable in $\mathcal X$, such that $\statefun{}(x) \in \mathit{range}(x)$ for any $x \in \mathcal{X}$,

\item $\measmap{}:  {\mathcal{V}}^{\mathcal X} \rightarrow  \mathcal S \rightarrow \distr{\mathcal{V}}$  is the \emph{measurement map}  that given a physical state
returns a function that associates to any sensor in $\mathcal S$ a discrete probability distribution over the set of possible sensed values; 
\end{itemize} 
\item a
\emph{logical (or cyber) component} 
$P$ that interacts with the sensors defined in $\xi$, and can communicate, via channels, with other cyber components. 
\end{itemize}
We write $\confCPS \xi  P$ to denote the resulting IoT system, and use 
$M$ and $N$ to range over IoT systems. 
\end{definition}

Let us now formalise the \emph{cyber component} of an IoT system.
Basically, we adapt Hennessy and Regan's \emph{timed process algebra TPL}~\cite{HR95}.

\begin{definition}[Logics]
\label{def:processes}
\emph{Logical components} of IoT systems  are defined by the
following grammar:\\[3pt]
\begin{math}
\begin{array}{rl}
P,Q \Bdf & \nil \q \big| \q 
\tick.P \q \big| \q 
P \parallel Q \q \big| \q 
\timeout {\mathit{pfx}.P} {Q}  \q \big| \q  
 H \langle \tilde{u} \rangle \q \big| \q
 \ifelse b P Q \q \big| \q P{\setminus}c  \\[3pt]
\mathit{pfx} \Bdf & \snda o v \Bor \rcva o z
\end{array}
\end{math}
\end{definition}
 The process $\tick.P$
sleeps for one time unit and then continues as $P$. 
We write $P \parallel Q$ to denote the \emph{parallel composition} of concurrent processes $P$ and $Q$. 
The process $\timeout {\mathit{pfx}.P} Q$ denotes
\emph{prefixing with timeout}. We recall that $o$ ranges over both channel and sensor names. Thus, for instance, 
$\timeout{\snda c v . P}Q$ sends the
value $v$ on channel $c$ and, after that, it continues as $P$; otherwise,
if no communication partner is available within one time unit, 
 it evolves into $Q$. The process $\timeout{\rcva c z.P}Q$ is the obvious counterpart for channel reception. 
On the other hand, the process    $\timeout{\rcva s z.P}Q$ reads the sensor $s$, according to 
the measurement map of the systems, and, after that, it
continues as $P$. 
The process  $\timeout{\snda s v . P}Q$  writes to the sensor $s$ and, after that, it continues as $P$; here, we wish to point out that this a \emph{malicious
activity}, as controllers may only access sensors for reading  sensed data. 
 Thus, the construct  $\timeout{\snda s v.P}Q$ serves to implement an 
\emph{integrity attack} that attempts at synchronising with the controller of sensor $s$ to 
provide a fake value $v$. 
In the following, we say that a process is \emph{honest} if it never writes on sensors. 
The definition of honesty  naturally lifts to IoT systems. 
In processes of the form $\tick.Q$ and $\timeout {\mathit{pfx}.P} Q$, the occurrence of $Q$ is said to be \emph{time-guarded}. 
\emph{Recursive processes} $H \langle \tilde{u} \rangle$ are defined via  equations $H(z_1,\ldots, z_k) = P$, where (i) the tuple $z_1,\ldots, z_k$ contains all the variables that appear free in $P$, and (ii) $P$ contains \emph{only  time-guarded occurrences} of the process identifiers, such as $H$ itself (to avoid \emph{zeno behaviours}). The two remaining constructs are standard; they model conditionals and channel restriction, respectively. 

Finally, we define how to compose IoT systems.  For simplicity,   we  compose two  systems only if they have the same physical environment. 
\begin{definition}[System composition]
\label{def:composing-systems}
Let $M_1 = \confCPS \xi  P_1$ and $M_2 = \confCPS \xi  P_2$ be two IoT systems, and $Q$ be a process whose sensors are defined 
in the physical environment $\xi$.  We write:
\begin{itemize}
\item  $M_1 \parallel M_2$ to denote $\confCPS \xi (P_1 \parallel P_2)$; 
\item $M_1 \parallel Q$ to denote $\confCPS \xi  { ({P_1}\parallel Q) }$; 
\item  $M_1{\setminus}c$ as an abbreviation for $\confCPS \xi  {({P_1} {\setminus}c)}$. 
\end{itemize}
\end{definition}

We conclude this section with the following  abbreviations that will be used in the rest of the paper. 
\begin{notation}
 We  write $P{\setminus}\{ c_1, c_2, \ldots , c_n \}$, or $P {\setminus}\tilde{c}$, to mean
$P{\setminus}{c_1}{\setminus}{c_2}\cdots{\setminus}{c_n}$. For simplicity, 
we sometimes abbreviate both $H(i)$ and $H \langle i \rangle$ with $H_i$. 
  We write $\mathit{pfx}.P$ as an abbreviation for 
the process defined via the  equation $\mathit{H} = 
\timeout{\mathit{pfx}.P}{\mathit{H}}$, where the process name $\mathit{H}$ does not occur in $P$. 
 We write $\tick^{k}.P$ as a shorthand for $\tick.\tick. \ldots \tick.P$, where the prefix $\tick$ appears $k \geq 0$ consecutive times.
 We write $\dummyN$ to denote a deadlocked IoT system  that  cannot perform any action.
\end{notation}


\subsection{Probabilistic labelled transition semantics}
\label{lab_sem}

\begin{table}[t]
\begin{displaymath}
\begin{array}{l@{\hspace*{5mm}}l}
\Txiom{Write}
{-}
{ { \timeout{\snda o v .P}Q } \trans{\snda o v}   P}
&
\Txiom{Read}
{-}   
{ { \timeout{\rcva o z .P}Q } \trans{\rcva o z}    {P}  }
\\[14pt]
\Txiom{Sync}
{ P \trans{\snda o v}  { P'}  \Q  Q \trans{\rcva o z}  { Q'} }
{ P \parallel  Q \trans{\tau}  {P'\parallel Q'{\subst v z}}}
&
\Txiom{Par}
{ P \trans{\lambda}  P' \Q \lambda \neq  \tick }
{ {P\parallel Q} \trans{\lambda} {P'\parallel Q}}
\\[14pt]

\Txiom{Res}{P \trans{\lambda} P' \Q \lambda \not\in \{ {\snda o v}, {\rcva o z} \}}{P {\setminus}o \trans{\lambda} {P'}{\setminus}o}

&

\Txiom{Rec}
{  P{\subst {\tilde{v}} {\tilde{z}}} \trans{\lambda}  Q \Q H(\tilde{z})=P}
{ H \langle \tilde{v} \rangle  \trans{\lambda}  Q}
\\[14pt] 

\Txiom{Then}{\bool{b}=\true \Q P \trans{\lambda} P'}
{\ifelse b P Q \trans{\lambda} P'}

&

\Txiom{Else}{\bool{b}=\false \Q Q \trans{\lambda} Q'}
{\ifelse b P Q \trans{\lambda} Q'}
\\[14pt]

\Txiom{TimeNil}{-}
{ \nil \trans{\tick}  \nil}

&
\Txiom{Delay}
{-}
{  { \tick.P} \trans{\tick}  P}
\\
[14pt]

\Txiom{Timeout}
{-}
{  {\timeout{\mathit{pfx}.P}{Q} }   \trans{\tick}  Q}
&
\Txiom{TimePar}
{
  P \trans{\tick}  {P'}  \Q 
   Q \trans{\tick} {Q'} 
}
{
  {P \parallel Q}   \trans{\tick}  { P' \parallel Q'}
}
\end{array}
\end{displaymath}
\caption{Labelled transition system for processes}
\label{tab:lts-P} 
\end{table}

As said before, sensors serve to observe the evolution of the 
physical state of an IoT system. 
However, sensors  are usually affected by an \emph{error/noise} that we represent
in our measurement maps by means of  discrete probability distributions. 
 For this reason,  we equip our calculus with a probabilistic labelled transition system. 
 In the following, the symbol $\epsilon$ ranges over distributions on physical environments, whereas $\pi$ ranges over distributions on (logical) processes. Thus, 
$\confCPS {\epsilon} {\pi}$ denotes the distribution over 
IoT systems  defined by $(\confCPS {\epsilon} {\pi})(\confCPS{\xi}{P})= {\epsilon}(\xi) \cdot \pi(P)$. The symbol $\gamma$ ranges over distributions on
IoT systems.

In \autoref{tab:lts-P}, we give a standard labelled transition system for logical components (timed processes), whereas in \autoref{tab:lts-S}
we rely on the LTS of \autoref{tab:lts-P} to define a simple 
pLTS  
for IoT  systems
by lifting transition rules from processes to systems.

In \autoref{tab:lts-P},  the meta-variable $\lambda$ ranges over labels in the set $\{\tau, \tick,  {\snda o v}, {\rcva o z} \}$. 
 Rule \rulename{Sync} serve to model synchronisation and value passing, on some name (for channel or sensor) $o$: if $o$ is a channel then we have standard point-to-point 
communication, whereas if $o$ is a sensor then 
this rule models an \emph{integrity attack} on sensor $s$, as the controller 
is provided with a fake value $v$. 
The remaining rules are standard. The symmetric counterparts of rules \rulename{Sync} 
and  \rulename{Par}  are  omitted.

According to \autoref{tab:lts-S}, IoT systems may fire four possible 
actions ranged over by $\alpha$. 
These actions represent: internal activities ($\tau$),  the
passage of time ($\tick$), 
channel transmission (${\out c v}$) and channel reception (${\inp c v}$). 

Rules \rulename{Snd} and \rulename{Rcv} model transmission and reception
on a channel $c$  with an external system, respectively.  
Rule \rulename{SensRead} models the reading of the value detected at a  \emph{sensor} $s$   according to the current physical environment $\xi = \langle \statefun{}, \measmap{} \rangle$. In particular, this rule says that if a process $P$ in a system 
$\confCPS \xi P$ 
reads a sensor $s$ defined in $\xi$ then it will get a value that may vary according to the probability distribution resulting by providing  the state function $\statefun{}$ and the sensor $s$ to the measurement map $\measmap{}$.

Rule \rulename{Tau} lifts internal actions from processes to
systems. This includes communications on channels and malicious accesses to 
 sensors' controllers. According to  Definition~\ref{def:composing-systems}, rule \rulename{Tau} models also channel communication between two parallel IoT systems sharing the same physical environment. 

A second lifting  occurs in rule
\rulename{Time} for timed actions $\tick$. 
Here, $\xi'$ denotes an 
admissible  physical environment for the next time slot,
nondeterministically chosen from the \emph{finite} set
 $\mathit{next}(\langle \statefun{} ,  \measmap{} \rangle)$. 
This set is defined as  $ 
  \{ \langle \statefun'{}, \measmap{} \rangle : \statefun'{}(x) \in \mathit{range}(x) \textrm{ for any } x \in \mathcal X\}$.\footnote{The finiteness follows 
from the finiteness of $\mathcal V$, and hence of $\mathit{range}(x)$, for any $x \in \mathcal X$.} As a consequence,  the rules in \autoref{tab:lts-S} define an \emph{image-finite} pLTS.  

For simplicity, we abstract from the \emph{physical process} behind our IoT systems.

\begin{table}[t]
\begin{displaymath}
\begin{array}{c}
\Txiom{Snd}
{P \trans{\snda c v}  P'  }
{\confCPS \xi  P   \trans{\out c v}   \confCPS {\dirac{\xi}}  {\dirac{P'} }}
\Q\Q\Q\Q
\Txiom{Rcv}
{P  \trans{\rcva c z}  P' }
{\confCPS \xi  P    \trans{\inp c v}  \confCPS {\dirac{\xi}}  {\dirac{P'{\subst v z}} }}
\\[17pt]

\Txiom{SensRead}{P \trans{\rcva s z} P'  \Q 
\mbox{\small{$\measmap{}(\statefun{})(s) = \sum_{i \in I} p_i \cdot \dirac{v_i}$}}
}
{\confCPS \xi  P \trans{\tau} \confCPS {\dirac{\xi}} {\sum_{i \in I}p_i \cdot  \dirac{P' \subst{v_i}{z}}}} 
\\[17pt]

\Txiom{Tau}{P \trans{\tau} P'  }
{ \confCPS \xi  P \trans{\tau} \confCPS {\dirac{\xi}}  {\dirac{P'}}}
\Q
\Txiom{Time}{ P \trans{\tick} {P'} \Q
\confCPS \xi  P \ntrans{\tau} \Q
 \xi' \in \mathit{next}(\xi) }
{\confCPS \xi  P \trans{\tick} \confCPS {\dirac{\xi'}}  {\dirac{P'}}}

\end{array}
\end{displaymath}
\caption{Probabilistic LTS for a IoT system $\confCPS \xi P$ with $\xi =
\langle \statefun{} , \measmap{} \rangle$}
\label{tab:lts-S} 
\end{table}

\section{Cyber-physical attacks on sensor devices}
\label{sec:cyber-physical-attackers}

In this section, we consider attacks tampering with sensors by eavesdropping and possibly  modifying the sensor measurements provided to the corresponding controllers. These attacks may affect both the \emph{integrity} and the \emph{availability} of the system under attack. We do not represent (well-known) attacks on communication channels as our focus is on attacks to physical devices and the consequent impact on the physical state. However, our technique can be easily generalised to deal with attacks on channels as well.  

\begin{definition}[Cyber-physical attack] A (pure) cyber-physical attack $A$ is a process derivable from the grammar of 
 \autoref{def:processes} such that:
\begin{itemize}
\item $A$ writes on at least one sensor;
\item $A$ never uses communication channels. 
\end{itemize}
\end{definition}

In order to make security assessments on our IoT systems, we adapt a
well-known approach called \emph{Generalized Non Deducibility on
Composition (GNDC)}~\cite{FM99}. 
Intuitively,  an attack $A$  affects an honest IoT system  $M$ if the execution of the composed system $M \parallel A$ differs from that of the original system $M$ in an observable manner. 
Basically, a cyber-physical attack can influence the system under attack in at least two different ways:
\begin{itemize}
\item The system $M \parallel  A$ might have non-genuine execution traces containing observables 
that cannot be reproduced by $M$; here the attack affects the \emph{integrity} of the system behaviour (\emph{integrity attack}).
\item The system $M$   might have execution traces containing
observables that cannot be reproduced by the system under attack 
$M \parallel  A$ (because they are prevented by the attack); this is an attack against the \emph{availability} of the system (\emph{DoS attack}). 
\end{itemize}

Now, everything is in place to provide a formal definition of \emph{system tolerance}
and \emph{system  vulnerability} with respect to a given attack. Intuitively, a system $M$ tolerates
an attack $A$ if the presence of the attack does not affect the behaviour of $M$; on the other hand $M$ is vulnerable to $A$ in a certain time interval if the
attack has an \emph{impact} on the behaviour of $M$ in that time interval.  

\begin{definition}[Attack tolerance]
\label{def:tolerance}
Let $M$ be a honest  IoT system. We say that $M$ 
 \emph{tolerates an attack $A$}  if  
$  M \parallel A  \bisimtoll{\infty}{0} M  $. 
\end{definition}

\begin{definition}[Attack vulnerability and impact]
\label{def:vulnerability}
Let $M$ be a honest  IoT system.
We say that $M$ is \emph{vulnerable to
 an attack $A$ in the time interval $m..n$  with \emph{impact} $p \in [0,1]$\/}, for $m\in \mathbb{N}^+$ and $n \in \mathbb{N}^+ \cup \{ \infty \} $, if 
 $m..n$ is the smallest time interval  such that:  
 (i) $ M \parallel A \bisimtoll{m-1}{0} M$, (ii)
$M \parallel A \bisimtoll{n}{p} M$, (iii) $M \parallel A \bisimtoll{\infty}{p} M$.\footnote{By \autoref{prop:tol-monotonicity}, at all time instants greater than $n$ the impact remains $p$.}
\end{definition}
Basically, the definition above says that if a system is vulnerable to an attack in the time interval $m..n$ then the perturbation introduced by the attack starts in the $m$-th time slot and reaches the maximum impact in the $n$-th time slot.

The following result says that both notions of tolerance and vulnerability 
are suitable for \emph{compositional reasonings\/}. More 
precisely, we prove that they are both preserved by parallel composition and 
channel restriction. Actually,  channel restriction may obviously make a system less vulnerable by hiding channels. 
\begin{theorem}[Compositionality] 
\label{thm:attack-tolerance-gen} 
Let $M_1 = \confCPS \xi  P_1$ and $M_2 = \confCPS \xi  P_2$ be two honest IoT systems with the same physical environment $\xi$, 
 $A$ an arbitrary attack, and $\tilde{c}$ a set of channels.
\begin{itemize}
\item
If both $M_1$ and $M_2$ tolerate $A$ then $(M_1 \parallel M_2) {\setminus} \tilde{c}$  
tolerates  $A$. 
\item
If $M_1$ is vulnerable to $A$ in the time interval $m_1..n_1$ with impact $p_1$,  and
$M_2$ is vulnerable to $A$ in the time interval $m_2..n_2$ with impact $p_2$,
then  $M_1 \parallel M_2$ is vulnerable to $A$ in a the time interval $\min(m_1,m_2)..\max(n_1,n_2)$ with an impact $p' \leq (p_1+p_2 - p_1 p_2)$. 
\item
If $M_1$ is vulnerable to $A$ in the interval $m_1..n_1$ with impact $p_1$ 
then  $M_1 {\setminus} \tilde{c}$ is vulnerable to $A$ in a time interval $m'..n' \subseteq m_1..n_1$ with an impact $p' \le p_1$. 
  \end{itemize}
\end{theorem}
%
Note that if an attack $A$ is tolerated by a system $M$ and can interact  with a honest process $P$ then the compound system $M  \parallel P$ may be  vulnerable to $A$. 
However, if $A$ does not write on the sensors of $P$ then it is tolerated by 
 $M \parallel P$  as  well. 
The bound $p' \leq (p_1+p_2 - p_1 p_2)$  can be explained as follows.
The likelihood that the attack does not impact on $M_i$ is $(1-p_i)$, for $i 
\in \{ 1,2 \}$.
Thus, the likelihood that the attack impacts neither on $M_1$ nor on $M_2$ is at least $(1-p_1)  (1-p_2)$.
Summarising, the likelihood that the attack impacts on at least one of the two systems $M_1$ and $M_2$ is at most $1- (1-p_1)  (1-p_2) = p_1+p_2 - p_1 p_2$.

An easy corollary of \autoref{thm:attack-tolerance-gen} allows us to lift 
the notions of tolerance and vulnerability from a honest system $M$ to the compound systems $M \parallel P$, for a honest process $P$. 
%
\begin{corollary}\label{thm:attack-tolerance}
Let $M$  be a honest system, $A$ an attack, $\tilde{c}$ a set of channels, and  $P$ a honest process that reads  sensors defined in $M$ but not those 
 written by $A$. 
\begin{itemize}
\item
If $M$ tolerates  $A$ then  $(M\parallel P) {\setminus} \tilde{c}$  
tolerates  $A$. 
\item
If  $M$ is vulnerable  to $A$   in the interval $m..n$  with impact $p$,  then  $(M\parallel P) {\setminus} \tilde{c}$ is vulnerable to $A$  in a time interval $m'..n' \subseteq m..n$,  with an impact $p' \leq p$. 
\end{itemize}
\end{corollary}

\section{Attacking  a smart surveillance system: A case study}
\label{sec:case}
Consider an alarmed ambient consisting of three rooms, $r_i$ for $i \in \{ 1, 2, 3 \}$, each of which equipped with a sensor $s_i$ to detect  unauthorised accesses. The alarm goes off if at least one of the three sensors
detects an intrusion.

The logics of the system can be easily specified in our language as follows:
\begin{displaymath}
{\small 
\begin{array}{rcl}
\mathit{Sys} &  =  & 
 \left( \mathit{Mng} \parallel \mathit{Ctrl_1}\parallel  \mathit{Ctrl_2} \parallel \mathit{Ctrl_3}\right){\setminus}\{c_1,c_2, c_3\}
\\[1pt]

\mathit{Mng} &  =  &  \rcva {c_1}{z_1}.\rcva {c_2}{z_2}.\rcva{c_3}{z_3} . 
\mathsf{if} \, 
(\bigvee_{i=1}^3  z_i{=} \mathsf{on}) \, \{ \snda{\mathit{alarm}}{\mathsf{on}} .\tick.\mathit{Check_{k}} \}  \, 
\mathsf{else} \,  \{ \tick.\mathit{Mng}\}
\\[1pt]

\mathit{Check_{0}} & = & \mathit{Mng}
\\[1pt]

\mathit{Check_{j}} & = &  \snda{\mathit{alarm}}{\mathsf{on}} . 
\rcva {c_1}{z_1}.\rcva {c_2}{z_2}.\rcva{c_3}{z_3} . 
\mathsf{if} \, 
(\bigvee_{i=1}^3  z_i= \mathsf{on}) \, \{ \mathit{\tick.Check_{k}} \} \: \\
&&
\mathsf{else} \:  \{  \tick.\mathit{Check_{j{-}1}}  \} \Q \textrm{for } j>0
\\[1pt]

\mathit{Ctrl_i} & = &  \rcva  {s_i} {z_i} .  \mathsf{if} \, 
(z_i{=}\mathsf{presence}) \, \{ \snda{c_i}{\mathsf{on}} .\tick. \mathit{Ctrl_i} \}  \, 
\mathsf{else} \,   \{ \snda{c_i}{\mathsf{off}} .\tick. \mathit{Ctrl_i} \}
\textrm{ for }i {\in} \{ 1, 2, 3 \}. \end{array}
}
\end{displaymath}

Intuitively, the process $\mathit{Sys}$ is composed by three controllers, $\mathit{Ctrl_i}$, 
one for each sensor $s_i$, and a manager  $\mathit{Mng}$ that interacts with the 
controllers via private channels $c_i$. The process $\mathit{Mng}$ fires an 
alarm if at least one of the controllers signals an intrusion. As usual in this kind of 
surveillance systems, the alarm will keep going off for $k$ instants of time after the last detected intrusion.

As regards the physical environment,  the physical state  $\statefun{} : \{ r_1, r_2, r_3 \} \rightarrow \{ \mathsf{presence} , \mathsf{absence} \} $ is set to $\statefun{}(r_i)=\mathsf{absence}$, for any $i \in \{ 1, 2, 3\}$. 
Furthermore, let $p_i^+$ and $p_i^-$  be the probabilities 
of having \emph{false  positives} (erroneously detected intrusion) and  \emph{false negatives} (erroneously missed intrusion) at sensor $s_i$\footnote{These probabilities are usually very 
 small; we assume them  smaller than $\frac{1}{2}$.}, respectively,  for 
 $i \in \{ 1 , 2, 3 \}$, 
the measurement function $\measmap{}$ is defined as follows: 
$ \measmap{}(\statefun{})(s_i)=(1{-}p_i^-) \, \dirac {\mathsf{presence}} + p_i^- \dirac {\mathsf{absence}}$, 
if $\statefun{}(r_i)=\mathsf{presence}$; 
$ \measmap{}(\statefun{})(s_i)=(1{-}p_i^+)\, \dirac {\mathsf{absence}} + p_i^+ \dirac {\mathsf{presence}}$, otherwise.

Thus, the whole IoT system has the form $\confCPS \xi {\mathit{Sys}}$, with $\xi = 
\langle \statefun{} , \measmap{} \rangle $.

 We start our analysis studying  the impact of  a simple cyber-physical attack that provides fake 
\emph{false positives} to the controller of one of the sensors $s_i$. 
This attack affects the \emph{integrity} of the system behaviour as the  system under attack will fire alarms  without any physical  intrusion.

\begin{example}[Introducing false positives] 
In this example, we provide an attack that tries to increase the number 
of false positives detected by the controller of some sensor $s_i$ during a specific time interval $m..n$, with $m, n \in \mathbb{N}$, $n \geq m > 0$. Intuitively, the attack 
 waits for $m-1$ time slots, then, during the time interval $m..n$, it provides  the controller of sensor $s_i$ with a fake intrusion signal. 
Formally,  
\begin{displaymath} 
\begin{array}{rcl}
  A_{\mathsf{fp}}(i,m,n) &  =  & \tick^{m-1} . B\langle i,  n-m+1 \rangle
\\[2pt]
 B(i, j) & = &  \ifelse {j = 0} {\nil} {\timeout{\snda {s_i} {\mathsf{presence}}  .  \tick. B \langle i , j-1 \rangle}  {B \langle i, j-1 \rangle}}
\, . 
\end{array}
\end{displaymath}
\end{example}

In the  following proposition, we use our metric to measure the perturbation introduced by the attack to the controller of a sensor $s_i$ by varying the time of observation of the system under attack. 

\begin{proposition}
\label{prop:case1}
Let $\xi$ be an arbitrary physical state for the systems $M_i = \confCPS \xi \mathit{Ctrl}_i$, for $i \in \{ 1, 2 , 3 \}$.  Then, 
\begin{itemize}
\item 
$ M_i \parallel  A_{\mathsf{fp}} \langle i , m, n \rangle \, \bisimtoll{j}{0} \,  M_i$, 
 for $j \in 1 .. m{-}1$;
\item
$ M_i \parallel  A_{\mathsf{fp}} \langle i , m, n \rangle \, \bisimtoll{j}{h} \,  M_i$, with    $h=1-(p_i^+)^{j-m+1} $,
	  for $j \in m  .. n $;
\item 
$ M_i \parallel  A_{\mathsf{fp}} \langle i , m, n \rangle \, \bisimtoll{j}{r} \,  M_i$, with    $r=1-(p_i^+)^{n-m+1} $,
	  for $j > n $ or $j=\infty$.  
\end{itemize}
\end{proposition} 
By an application of \autoref{def:vulnerability} we can measure the impact of the 
attack $A_{\mathsf{fp}}$ to the  (sub)systems $ \confCPS {\xi } { \mathit{Ctrl_i}}$. 
\begin{corollary}
The IoT systems
 $ \confCPS {\xi } { \mathit{Ctrl_i}}$ are vulnerable to the 
   attack $ A_{\mathsf{fp}} \langle i , m, n \rangle $ in the time interval $m..n$  with impact  $   1-(p_i^+)^{  n - m +1 } $.
 \end{corollary}
Note that the vulnerability window $m..n$ coincides with the activity period of the attack $A_{\mathsf{fp}}$. This means that the system under attack recovers its normal behaviour immediately after the termination of the attack. However, in general, an attack may impact the  behaviour of the target system long after its termination. 

Note also that the attack $ A_{\mathsf{fp}}\langle i , m, n \rangle$ has an impact not only on the controller $\mathit{Ctrl}_i$  but also on the whole system $\confCPS \xi \mathit{Sys}$. This because the process $\mathit{Mng}$ will surely fire the alarm as it will receive at least one
intrusion detection from $\mathit{Ctrl}_i$.  However, by an 
application of \autoref{thm:attack-tolerance} we can prove that the impact on the whole system  will not get amplified. 
\begin{proposition}[Impact of the attack $ A_{\mathsf{fp}}$]
The system $ \confCPS \xi  \mathit{Sys} $ is vulnerable  to the 
   attack $ A_{\mathsf{fp}} \langle i , m, n \rangle $ in a time interval $m'..n' \subseteq m..n$ with impact $p' \leq   1-(p_i^+)^{  n - m +1}$.
 \end{proposition}

Now, the reader may wonder what happens if we consider a complementary attack that provides fake 
\emph{false negatives} to the controller of one of the sensors $s_i$. In this
case, the attack affects  the \emph{availability} of the system behaviour as the system will no fire the alarm   in the presence of a real intrusion. This 
because a real intrusion  will be somehow ``hidden'' by the attack. 
\begin{example}[Introducing false negatives]
The goal of the following attack is to increase the number of 
false negatives during the time interval $m..n$, with $n \geq m > 0$. Formally, 
the attack is defined as follows: 
\begin{displaymath} 
\begin{array}{rcl}
  A_{\mathsf{fn}}(i,m,n) &  =  & \tick^{m-1} . C \langle i,  n-m+1 \rangle
\\[2pt]
 C(i, j) & = &  \ifelse {j = 0} {\nil} {\timeout{\snda {s_i} {\mathsf{absence}}  . \tick. C \langle i , j-1 \rangle}  {C \langle i, j-1 \rangle}}
\, . 
\end{array}
\end{displaymath}
\end{example}
In the  following proposition, we use our metric to measure the deviation introduced by the attack $ A_{\mathsf{fn}}$ to the controller of a sensor $s_i$. With 
no surprise we get a result that is the symmetric version of  \autoref{prop:case1}.

\begin{proposition}
\label{prop:case2}
Let $\xi$ be an arbitrary physical state for the system $M_i = \confCPS \xi \mathit{Ctrl}_i$, for $i \in \{ 1, 2 , 3 \}$. Then, 
\begin{itemize}
\item 
$ M_i \parallel  A_{\mathsf{fn}} \langle i , m, n \rangle \, \bisimtoll{j}{0} \,  M_i$, 
 for $j \in 1 .. m{-}1$;
\item         	
$ M_i \parallel  A_{\mathsf{fn}} \langle i , m, n \rangle \, \bisimtoll{j}{h} \,  M_i$, with    $h=1-(p_i^-)^{j-m+1} $,
	  for $j \in m  .. n $;
\item 
$ M_i \parallel  A_{\mathsf{fn}} \langle i , m, n \rangle \, \bisimtoll{j}{r} \,  M_i$, with    $r=1-(p_i^-)^{n-m+1} $,
	  for $j > n $ or $j=\infty$. 
\end{itemize}
\end{proposition}

Again, by an application of \autoref{def:vulnerability} we can measure the impact of the 
attack $A_{\mathsf{fn}}$ to the  (sub)systems $ \confCPS {\xi } { \mathit{Ctrl_i}}$. 
\begin{corollary}
The IoT systems
 $ \confCPS {\xi } { \mathit{Ctrl_i}}$ are vulnerable to the 
   attack $ A_{\mathsf{fn}} \langle i , m, n \rangle $  in the time interval $m..n$  with impact  $   1-(p_i^-)^{  n - m +1 } $.
 \end{corollary}

As our timed  metric is compositional, by an application of  \autoref{thm:attack-tolerance} we can estimate the impact of the 
attack $A_{\mathsf{fn}}$ to the whole system $\confCPS \xi \mathit{Sys}$. 
\begin{proposition}[Impact of the attack $ A_{\mathsf{fn}}$]
The system $ \confCPS \xi  \mathit{Sys} $ is vulnerable to the 
   attack $ A_{\mathsf{fn}} \langle i , m, n \rangle $  in a time interval
 $m'..n' \subseteq m..n$  with impact $ p' \leq  1-(p_i^-)^{  n - m +1}$.
 \end{proposition}


\section{Conclusions, related and future work}
\label{sec:conclusions}

We have proposed a timed  generalisation of the 
$n$-bisimulation metric~\cite{vB12}, called \emph{timed bisimulation metric}, 
obtained by defining two functionals over the complete lattice of the functions assigning a distance in $[0,1]$ to each pair of systems: the former deals with the distance accumulated when executing untimed  steps, the latter with the distance introduced by  timed actions.

We have used our timed bisimulation metrics to provide a formal and \emph{compositional} notion of \emph{impact metric} for \emph{cyber-physical attacks} on \emph{IoT systems} specified in a simple timed process calculus. In particular, we have focussed on 
cyber-physical attacks targeting sensor devices (attack on sensors are by far the most studied cyber-physical attacks~\cite{survey-CPS-security-2016}). We have used our timed weak bisimulation with tolerance to formalise the notions of  \emph{attack tolerance} 
and \emph{attack vulnerability with a given impact $p$}. 
In particular, 
a system  $M$ is said to be vulnerable to an attack $A$ in the time interval $m..n$ with impact $p$ if 
the perturbation introduced by $A$ becomes observable in the $m$-th time slot and yields the maximum impact $p$ in the $n$-th time slot. Here, we wish to stress that  the \emph{vulnerability window} $m..n$ is quite informative. 
 In practise, this interval says when an attack will produce observable effects on the system under attack. 
Thus,  if $n$ is finite  we have an attack with \emph{temporary effects}, otherwise we have an attack with \emph{permanent effects}.  Furthermore,  if the attack is quick enough, and terminates  well before the time instant $m$,  then we
have a \emph{stealthy attack} that affects the system late enough to allow
\emph{attack camouflages}~\cite{GGIKLW2015}. On the other hand, if 
at time $m$ the attack is far from termination, then the IoT system under attack 
has good chances of undertaking countermeasures to stop the attack.

As a case study,  we have estimated the impact of two  cyber-physical attacks on sensors  that introduce \emph{false positives} and \emph{false negatives}, respectively, into a  simple surveillance system, affecting the
 \emph{integrity} and the \emph{availability} of the IoT system. Although our attacks are quite simple, the specification language and the corresponding metric semantics presented in the paper allow us to deal with smarter attacks, such as \emph{periodic attacks} with constant or variable period of attack. Moreover, we can easily extend our threat model to recover (well-known) attacks on communication channels.

\subsubsection*{Related work.}
We are aware of a  number of works using formal methods for \CPS{} security, although they apply methods, and most of the time  have goals, that are quite different from ours.

Burmester et al.~\cite{BuMaCh2012} employed \emph{hybrid timed automata} to give a threat model 
based on the traditional Byzantine fault model for
crypto-security.
However, as remarked in~\cite{TeShSaJo2015}, cyber-physical attacks and faults have inherently distinct
characteristics.  
In fact, unlike faults, cyber-physical attacks may be performed over a significant number of attack points and in a coordinated way.

In~\cite{Vig2012}, Vigo presented an attack scenario that addresses some
of the peculiarities of a cyber-physical adversary, and discussed how this
scenario relates to other attack models popular in the security protocol
literature. Then, in~\cite{Vigo2015,VNN2013} Vigo et al.\ proposed an
untimed calculus of broadcasting processes 
equipped with notions of 
failed and unwanted communication. They focus on DoS attacks without taking into consideration timing aspects or attack impact.

 Bodei et al.~\cite{BDFG16,BDFG17}  proposed an untimed 
process calculus, IoT-LySa, supporting  a control flow analysis that safely approximates the abstract behaviour of IoT systems. Essentially, they track how data spread from sensors to the logics of the network, 
and how physical data are manipulated. 

  Rocchetto and
Tippenhaur~\cite{RocchettoTippenhauer2016a} introduced a taxonomy of the diverse attacker models proposed
for \CPS{} security and outline requirements for generalised attacker
models; in~\cite{RocchettoTippenhauer2016b}, they then proposed an
extended Dolev-Yao attacker model suitable for \CPS{s}. 
In their approach, physical layer interactions are modelled as abstract
interactions between logical components to support reasoning on the
physical-layer security of \CPS{s}. This is done by introducing additional
orthogonal channels. Time is not represented.

Nigam et al.~\cite{Nigam-Esorics2016} 
worked around the notion of Timed Dolev-Yao Intruder Models for
Cyber-Physical Security Protocols by bounding the number of intruders
required for the automated verification of such protocols. Following a
tradition in security protocol analysis, they provide an answer to the
question: How many intruders are enough for verification and where should
they be placed? 
Their notion of time is somehow different from ours, as they focus on the time a
message needs to travel from an agent to another. The paper does not
mention physical devices, such as sensors and/or actuators.

%
%
Finally, Lanotte et al.~\cite{paperCSF2017} defined a hybrid process calculus to model both \CPS{s} and cyber-physical attacks; they defined a threat model for cyber-physical attacks to physical devices and provided a proof methods to assess attack tolerance/vulnerability with respect to a timed trace semantics (no tolerance allowed).

\subsubsection*{Future work.}
Recent works~\cite{LM11,GLT16,LMT17b,LMT17,GT18} have shown that bisimulation metrics are suitable 
for compositional reasoning, as the distance between two complex systems can be often derived in terms of the distance between their components. In this respect, \autoref{thm:attack-tolerance-gen} and \autoref{thm:attack-tolerance} allows us compositional reasonings when computing the impact of attacks 
on a target system, 
in terms of the impact on its sub-systems. We believe that this result can be  generalised to estimate the impact of parallel attacks of the form $A = A_1 \parallel \ldots \parallel A_k$
in terms of the impacts of each malicious module $A_i$.

As future work, we also intend to adopt our impact metric  in more involved languages for \emph{cyber-physical systems and attacks}, such as the language developed in~\cite{paperCSF2017}, 
with an explicit representation of  physical processes via differential equations or their discrete counterpart, difference equations. 

\paragraph*{Acknowledgements.}
We thank the anonymous reviewers for  valuable comments.  This work has been partially supported by the project ``Dipartimenti di Eccellenza 2018-2022'', funded by the Italian Ministry of Education, Universities and Research (MIUR), and by the Joint Project 2017 ``Security Static Analysis for Android Things'', funded by the University of Verona and JuliaSoft Srl.  \textcolor{white}{\cite{CHM15,MKN02}}


\bibliographystyle{abbrv}
\bibliography{IoT_bib}



\appendix
\input{SecLMT-Proof.tex}

\end{document}

\appendix

\section{Proofs}

\subsection{Proofs of \S~\ref{sec:calculus}}

In order to prove Proposition~\ref{prop:sys} and Proposition~\ref{prop:X}, we use the following lemma that formalises the invariant properties binding 
the state variable $\mathit{temp}$ with the activity of the cooling system.

Intuitively, when the cooling system is inactive the value of the state
variable $\mathit{temp}$ lays in the real interval $[0, 11.5]$.
Furthermore, if the coolant is not active and the variable $\mathit{temp}$
lays in the real interval $(10.1, 11.5]$, then the cooling will be turned
on in the next time slot. Finally, when active the cooling system will
remain so for $k\in1..5$ time slots (counting also the current time slot)
with the variable $\mathit{temp}$ being in the real interval $(
9.9-k{*}(1{+}\delta) , 11.5-k{*}(1{-}\delta)]$.

\begin{lemma} 
\label{lem:sys}
Let $\mathit{Sys}$ be the system defined in Example~\ref{exa:sys}.
Let
\begin{small}
\begin{displaymath}
\mathit{Sys} = \mathit{Sys_1} \trans{t_1}\trans\tick 
\mathit{Sys_2}\trans{t_2}\trans\tick  \dots 
\trans{t_{n-1}}\trans\tick  \mathit{Sys_n}
\end{displaymath}
\end{small}%
such that the traces $t_j$ contain no $\tick$-actions, for any $j \in  1 .. n{-}1 $, and for any  $i \in  1 .. n $, $\mathit{Sys_i}= \confCPS {E_i}{P_i} $ with 
$E_i = \envCPS 
{\statefun^i{}} 
{\actuatorfun^i{}} 
{ \delta }  
{\evolmap{}}
{ \epsilon }  
{\measmap{}}   
{\invariantfun{}}$.
Then, for any $i \in 1 .. n{-}1 $, we have the following:
\begin{enumerate}

\item \label{uno}
 if   $ \actuatorfun^i{}(\mathit{cool})= \off $ then
 $\statefun^i{}(\mathit{temp})  \in [0, 11.1+\delta ]$;
with  $\statefun^i{}(\mathit{stress})=0$ if $ \statefun^i{}(\mathit{temp})  \in [0, 10.9+\delta ] $, and  $\statefun^i{}(\mathit{stress})=1 $, otherwise; 

\item \label{due}
  if   $ \actuatorfun^i{}(\mathit{cool})= \off $ and 
$\statefun^i{}(\mathit{temp})\in (10.1, 11.1+\delta ]$ then, in the next time slot,  $\actuatorfun^{i{+}1}{}(\mathit{cool})=\on$
  and $\statefun^{i{+}1}{}(\mathit{stress})  \in 1..2$;

\item \label{tre}
 if  $ \actuatorfun^i{}(\mathit{cool})=\on$ then   $\statefun^i{}(\mathit{temp}) \in ( 9.9-k {*}(1{+}\delta) , 11.1+\delta -k{*}(1{-}\delta)] $, 
for some  $k  \in 1 .. 5 $   such that $\actuatorfun^{i-k}{}(\mathit{cool})=\off $ and 
$\actuatorfun^{i-j}{}(\mathit{cool}) =\on $, for $j \in 0..k{-}1$;
moreover,  if  $k\in 1.. 3$ then     $\statefun^i{}(\mathit{stress})  \in  1..k{+}1  $, otherwise,  
$\statefun^i{}(\mathit{stress}) =0$. 
\end{enumerate}
\end{lemma}
\begin{proof}
Let us write $v_i$ and $s_i$ to denote the values of the state variables
$\mathit{temp}$ and $\mathit{stress}$, respectively, in the systems
$\mathit{Sys_i}$, i.e., $\statefun^i{} (\mathit{temp})=v_i $ and
$\statefun^i{} (\mathit{stress})=s_i $. Moreover, we will say that the
coolant is active (resp., is not active) in $\mathit{Sys_i}$ if
$\actuatorfun^i{}(\mathit{cool})=\on$ (resp.,
$\actuatorfun^i{}(\mathit{cool})=\off$).

The proof is by mathematical induction on $n$, i.e., the number of
$\tick$-actions of our traces.

The \emph{case base} $n=1$ follows directly from the definition of $\mathit{Sys}$. 

Let us prove the \emph{inductive case}. 
We assume that the three statements hold for $n-1$ and prove that they  
also hold for $n$.
\begin{enumerate}[noitemsep]
\item Let us assume that the cooling  is not active  in $\mathit{Sys_{n}}$.
In this case, we prove that $v_n \in [0, 11.1+\delta ]$, with and $s_n=0$ if $ v_n  \in [0, 10.9+\delta ] $, and $s_n=1$ otherwise.

We consider separately the cases in which the coolant is active or not in $\mathit{Sys_{n-1}}$
\begin{itemize}[noitemsep]
\item Suppose the coolant is not active in $\mathit{Sys_{n{-}1}}$ (and
not active in $\mathit{Sys_{n}}$).

By the induction hypothesis we have 
$v_{n-1} \in [0, 11.1+\delta ]$; with $s_{n{-}1}=0$ if $ v_{n{-}1}  \in [0, 10.9+\delta ] $, and $s_{n{-}1}=1 $ otherwise. Furthermore,  if   
$v_{n-1} \in (10.1, 11.1+\delta ]$, then, by the induction hypothesis, the coolant must be active in $\mathit{Sys_{n}}$.
Since we know that in $\mathit{Sys_n}$ the cooling is not active,
it follows that $v_{n-1} \in [0, 10.1]$ and $s_n =0$.
Furthermore, in $\mathit{Sys_{n}}$ the temperature
will increase of a value laying in the real interval $[1-\delta,1+\delta]=[0.6,1.4]$. Thus, $v_{n}$ will be in 
$ [0.6, 11.1+\delta ]\subseteq[0, 11.1+\delta ]$.  
Moreover, if $v_{n-1} \in [0, 9.9]$, then the state variable $\mathit{stress}$ is not incremented and hence $s_n=0$ with   
$ v_n  \in [0+1-\delta \, , \,  9.9+ 1+\delta ]=[0.6 \, , \, 10.9+\delta]\subseteq [0 \, , \, 10.9+\delta ] $. Otherwise, 
if $v_{n-1} \in (9.9,10.1]$, then the state variable $\mathit{stress}$ is incremented, and hence $s_n=1$.

\item Suppose the coolant is active in $\mathit{Sys_{n{-}1}}$ (and not
active in $\mathit{Sys_{n}}$).

By the induction hypothesis, $v_{n-1} \in ( 9.9-k *(1+\delta) , 11.1+\delta
-k*(1-\delta)]$ for some $k \in 1..5$ such that the coolant is not active
in $\mathit{Sys_{n{-}1{-}k}}$ and is active in $\mathit{Sys_{n{-}k}},
\ldots, \mathit{Sys_{n-1}}$.

The case $k \in \{1,\ldots,4\}$  is not admissible. 
In fact if  $k \in \{1,\ldots,4\}$ then the coolant would be active for less than $5$ $\tick$-actions as we know that 
$\mathit{Sys_{n}}$ is not active. 
Hence, it must be $k=5$. Since $\delta=0.4$ and $k=5$, it holds that $v_{n-1 }\in (9.9-5*1.4, 11.1+0.4 -5*0.6]=(2.8, 8.6] $
and $s_{n{-}1} =0$. Moreover, since
the coolant is active for $5$ time slots, in $\mathit{Sys_{n{-}1}}$ the controller and the $\mathit{IDS}$ synchronise together via channel $\mathit{sync}$ and hence the $\mathit{IDS}$ checks the
temperature. Since $v_{n-1} \in (2.8, 8.6]$ the $\mathit{IDS}$ process 
sends to the controller a command to 
$\mathsf{stop}$ the cooling, and the controller will switch off the 
cooling system. Thus, in the next time slot, the temperature
will increase of a value laying in the real interval $[1-\delta,1+\delta]=[0.6,1.4]$. As
a consequence, in $\mathit{Sys_{n}}$ we will have $v_{n}
\in [2.8+0,6, 8.6+1.4]=[3.4,10] \subseteq [0, 11.1+\delta ]$.
Moreover, since $v_{n-1} \in (2.8, 8.6]$ and $s_{n{-}1} =0$, we derive that the state variable $\mathit{stress}$ is not increased and hence
$s_{n } =0$, with $v_{n} \in [3.4,10] \subseteq [0, 10.9+\delta ]$. 
\end{itemize}

\item Let us assume that the coolant is not active in $\mathit{Sys_{n}}$
and $v_n \in (10.1, 11.1+\delta ]$; we prove that the coolant is active in
$\mathit{Sys_{n{+}1}}$ with $s_{n {+}1} \in 1..2 $. Since the coolant is
not active in $\mathit{Sys_{n}}$, then it will check the temperature
before the next time slot. Since $v_n \in (10.1, 11.1+\delta ]$ and
$\epsilon=0.1$, then the process $\mathit{Ctrl}$ will sense a temperature
greater than $10$ and the coolant will be turned on. Thus, the coolant
will be active in $\mathit{Sys_{n{+}1}}$. Moreover, since $v_n \in (10.1,
11.1+\delta ]$, and $s_{n}$ could be either $0$ or $1$, the state variable
$\mathit{stress}$ is increased and therefore $s_{n {+}1} \in 1..2$.

\item Let us assume that the coolant is active in $\mathit{Sys_{n}}$; we
prove that $v_{n} \in ( 9.9-k *(1+\delta), 11.1+\delta -k*(1-\delta)] $
for some $k \in 1..5 $ and the coolant is not active in
$\mathit{Sys_{n{-}k}}$ and active in $\mathit{Sys_{n-k+1}}, \dots,
\mathit{Sys_{n}}$. Moreover, we have to prove that if $k\leq 3$ then $s_n
\in 1..k{+}1 $, otherwise, if $k > 3$ then $s_n =0$.

We prove the first statement. That is, we prove that $v_{n} \in ( 9.9-k
*(1+\delta), 11.1+\delta -k*(1-\delta)] $, for some $k \in 1..5 $, and the
coolant is not active in $\mathit{Sys_{n{-}k}}$, whereas it is active in
the systems $\mathit{Sys_{n-k+1}}, \dots, \mathit{Sys_{n}}$.


We separate the case in which the coolant is active in
$\mathit{Sys_{n{-}1}}$ from that in which is not active.

\begin{itemize}[noitemsep]
\item Suppose the coolant is not active in $\mathit{Sys_{n{-}1}}$ (and active in $\mathit{Sys_{n}}$).

In this case $k=1$ as the coolant is not active in $\mathit{Sys_{n-1}}$
and it is active in $\mathit{Sys_{n}}$. Since $k=1$, we have to prove $v_n
\in (9.9-(1+\delta), 11.1+\delta-(1-\delta)]$.

However, since the coolant is not active in $\mathit{Sys_{n-1}}$ and is
active in $\mathit{Sys_{n}}$ it means that the coolant has been switched
on in $\mathit{Sys_{n-1}}$ because the sensed temperature was above $10$
(since $\epsilon=0.1$ this may happen only if $v_{n-1} > 9.9$). By the
induction hypothesis, since the coolant is not active in
$\mathit{Sys_{n-1}}$, we have that $v_{n-1} \in [0, 11.1+\delta ]$.
Therefore, from $v_{n-1} > 9.9$ and $v_{n-1} \in [0, 11.1+\delta ]$ it
follows that $v_{n-1} \in (9.9, 11.1+\delta ]$. Furthermore, since the
coolant is active in $\mathit{Sys_{n}}$, the temperature will decrease of
a value in $[1-\delta,1+\delta]$ and therefore $v_n \in (9.9-(1+\delta),
11.1+\delta-(1-\delta)]$, which concludes this case of the proof.

\item Suppose the coolant is active in $\mathit{Sys_{n{-}1}}$ (and active
in $\mathit{Sys_{n}}$ as well).

By the induction hypothesis, there is $h \in 1..5$ such that $v_{n-1} \in
( 9.9-h *(1+\delta) , 11.1+\delta -h*(1-\delta)] $ and the coolant is not
active in $\mathit{Sys_{n{-}1{-}h}}$ and is active in
$\mathit{Sys_{n{-}h}}, \ldots, \mathit{Sys_{n{-}1}}$.

The case $h=5$ is not admissible. In fact, since $\delta=0.4$, if $h=5$
then $v_{n-1 }\in (9.9-5*1.4, 11.1+\delta -5*0.6]=(2.8, 8.6] $.
Furthermore, since the cooling system has been active for $5$ time
instants, in $\mathit{Sys_{n{-}1}}$ the controller and the IDS synchronise
together via channel $\mathit{sync}$, and the $\mathit{IDS}$ checks the
received temperature. As $v_{n-1 }\in (2.8, 8.6] $, the $\mathit{IDS}$
sends to the controller via channel $\mathit{ins}$ the command
$\mathsf{stop}$. This implies that the controller should turn off the
cooling system, in contradiction with the hypothesis that the coolant is
active in $\mathit{Sys_{n }}$.

Hence, it must be $h \in 1 .. 4$. Let us prove that for $k=h+1$ we obtain
our result. Namely, we have to prove that, for $k=h+1$, (i) $v_{n} \in (
9.9-k *(1+\delta), 11.1+\delta -k*(1-\delta)] $, and (ii) the coolant is
not active in $\mathit{Sys_{n{-}k}}$ and active in $\mathit{Sys_{n-k+1}},
\dots, \mathit{Sys_{n}}$.

Let us prove the statement (i). By the induction hypothesis, it holds that
$v_{n-1} \in ( 9.9-h *(1+\delta) , 11.1+\delta -h*(1-\delta)] $. Since the
coolant is active in $\mathit{Sys_{n}}$, the temperature will decrease
Hence, $v_{n } \in ( 9.9-(h+1) *(1+\delta) , 11.1+\delta
-(h+1)*(1-\delta)] $. Therefore, since $k=h+1$, we have that $v_{n} \in (
9.9-k *(1+\delta) , 11.1+\delta -k*(1-\delta)] $.

Let us prove the statement (ii). By the induction hypothesis the coolant
is not active in $\mathit{Sys_{n-1-h}}$ and it is active in
$\mathit{Sys_{n-h}}, \ldots, \mathit{Sys_{n-1}}$. Now, since the coolant
is active in $\mathit{Sys_{n}}$, for $k=h+1$, we have that the coolant is
not active in $\mathit{Sys_{n-k}}$ and is active in $\mathit{Sys_{n-k+1}},
\ldots, \mathit{Sys_{n}}$, which concludes this case of the proof.
\end{itemize}

Thus, we have proved that $v_{n} \in ( 9.9-k *(1+\delta), 11.1+\delta
-k*(1-\delta)] $, for some $k \in 1..5 $; moreover, the coolant is not
active in $\mathit{Sys_{n{-}k}}$ and active in the systems
$\mathit{Sys_{n-k+1}}, \dots, \mathit{Sys_{n}}$.

It remains to prove that $s_n \in 1..k{+}1 $ if $k\leq 3$, and $s_n =0$,
otherwise.

By inductive hypothesis, since the coolant is not active in
$\mathit{Sys_{n{-}k}}$, we have that $s_{n{-}k} \in 0..1$. Now, for $k \in
[1..2]$, the temperature could be greater than $9.9$. Hence if the state
variable $\mathit{stress}$ is either increased or reset, then $s_n \in
1..k{+}1 $, for $k\in 1.. 3$. Moreover, since for $k\in 3..5 $ the
temperature is below $9.9$, it follows that $s_n =0$ for $k> 3$.
\end{enumerate}
\end{proof}

\begin{proof}[Proof of Proposition~\ref{prop:sys}]
Since $\delta=0.4$, by Lemma~\ref{lem:sys} the value of the state variable
$\mathit{temp}$ is always in the real interval $[0, 11.5]$. As a
consequence, the invariant of the system is never violated and the system
never deadlocks. Moreover, after $5$ time units of cooling, the state
variable $\mathit{temp}$ is always in the real interval $( 9.9-5 *1.4 ,
11.1+0.4-5*0.6]=(2.9, 8.5]$. Hence, the process $\mathit{IDS}$ will never
transmit on the channel $\mathit{alarm}$.

Finally, by Lemma~\ref{lem:sys} the maximum value reached by the state variable $\mathit{stress}$ is $4$ and therefore the system does not reach unsafe states.
\end{proof}

\begin{proof}[Proof of  Proposition~\ref{prop:X}]
Let us prove the two statements separately. 
\begin{itemize}
\item Since $\epsilon=0.1$, if process $\mathit{Ctrl}$ senses a
temperature above $10$ (and hence $\mathit{Sys}$ turns on the cooling)
then the value of the state variable $\mathit{temp}$ is greater than
$9.9$. By Lemma~\ref{lem:sys}, the value of the state variable
$\mathit{temp}$ is always less than or equal to $11.1+\delta $. Therefore, if
$\mathit{Ctrl}$ senses a temperature above $10$, then the value of the
state variable $\mathit{temp}$ is in $(9.9,11.1+\delta ]$.

\item By Lemma~\ref{lem:sys} (third item), the coolant can be active for no
more than $5$ time slots. Hence, by Lemma~\ref{lem:sys}, when $\mathit{Sys}$
turns off the cooling system the state variable $\mathit{temp}$ ranges
over $( 9.9-5 *(1+\delta) , 11.1+\delta-5*(1-\delta)]$. 
\end{itemize}
\end{proof}

\subsection{Proofs of \S~\ref{sec:cyber-physical-attackers}}

%
%
%

\begin{proof}[Proof of Proposition~\ref{prop:att:DoS}]
We distinguish the two cases, depending on $m$.
\begin{itemize}[noitemsep]
\item Let $m\leq 8$.
We recall that the cooling system is activated only when the sensed temperature is above $10$. Since $\epsilon = 0.1$, when this happens the state variable $\mathit{temp}$ must be at least $9.9$. Note that after $m{-}1 \leq 7$ $\tick$-actions, when the attack tries to interact with the
controller of the actuator $\mathit{cool}$, the variable $\mathit{temp}$ may reach at most $7* (1 + \delta)= 7 * 1.4=9.8$ degrees. Thus, the cooling system will not be activated and the attack will not have any effect.

\item Let $m>8$. 
By Proposition~\ref{prop:sys}, the system $\mathit{Sys}$ in
isolation may never deadlock, it does not get into an unsafe state, and it may never emit an output on channel $\mathit{alarm}$. Thus, any execution trace of the system $\mathit{Sys}$ consists of a sequence of $\tau$-actions and $\tick$-actions.

In order to prove the statement it is enough to show the following four facts:
\begin{itemize}[noitemsep]
\item the system $\mathit{Sys} \parallel A_m$ may not deadlock in the first
$m+3$ time slots; 

\item the system $\mathit{Sys} \parallel A_m$ may not emit any output in the first $m+3$ time slots; 

\item the system $\mathit{Sys} \parallel A_m$ may not enter in an unsafe
state in the first $m+3$ time slots;

\item the system $\mathit{Sys} \parallel A_m$ has a trace reaching un unsafe state from the $(m{+}4)$-th time slot on, and until the invariant gets violated and the system deadlocks. 
\end{itemize}


The first three facts are easy to show as the attack may steal the command
addressed to the actuator $\mathit{cool}$ only in the $m$-th time slot.
Thus, until time slot $m$, the whole system behaves correctly. In
particular, by Proposition~\ref{prop:sys} and Proposition~\ref{prop:X}, no alarms,
deadlocks or violations of safety conditions occur, and the temperature
lies in the expected ranges. Any of those three actions requires at least
further $4$ time slots to occur. Indeed, by Lemma~\ref{lem:sys}, when the
cooling is switched on in the time slot $m$, the variable
$\mathit{stress}$ might be equal to $2$ and hence the system might not
enters in an unsafe state in the first $m+3$ time slots. Moreover, an
alarm or a deadlock needs more than $3$ time slots and hence no alarm can
occur in the first $m+3$ time slots.

Let us show the fourth fact, i.e., that there is a trace where the system
$\mathit{Sys} \parallel A_m$ enters into an unsafe state starting from the
$(m{+}4)$-th time slot and until the invariant gets violated.

Firstly, we prove that for all time slots $n$, with $9\leq n < m$, there
is a trace of the system $\mathit{Sys} \parallel A_m$ in which the state
variable $\mathit{temp}$ reaches the values $10.1$ in the time slot $n$.

The fastest trace reaching the temperature of $10.1$ degrees requires
$\lceil \frac{10.1}{1 + \delta }\rceil = \lceil \frac{10.1}{1.4 }\rceil
=8$ time units, whereas the slowest one $\lceil \frac{10.1}{1 - \delta
}\rceil = \lceil \frac{10.1}{0.6 }\rceil =17$ time units. Thus, for any
time slot $n$, with $9 \leq n \leq 18$, there is a trace of the system
where the value of the state variable $\mathit{temp}$ is $10.1$. Now, for
any of those time slots $n$ there is a trace in which the state variable
$\mathit{temp}$ is equal to $10.1$ in all time slots $n+10i < m$, with
$i\in \mathbb{N}$. Indeed, when the variable $\mathit{temp}$ is equal to
$10.1$ the cooling might be activated. Thus, there is a trace in which the
cooling system is activated. We can always assume that during the cooling
the temperature decreases of $1+\delta$ degrees per time unit, reaching at
the end of the cooling cycle the value of $5$. This entails that the trace
may continue with $5$ time slots in which the variable $\mathit{temp}$ is
increased of $1+\delta$ degrees per time unit; reaching again the value
$10.1$. Thus, for all time slots $n$, with $9 \leq n < m$, there is a
trace of the system $\mathit{Sys} \parallel A_m$ in which the state
variable $\mathit{temp}$ is $10.1$ in $n$.

As a consequence, we can suppose that in the $m{-}1$-th time slot there is
a trace in which the value of the variable $\mathit{temp}$ is $10.1$.
Since $\epsilon=0.1$, the sensed temperature lays in the real interval
$[10,10.2]$. Let us focus on the trace in which the sensed temperature is
$10$ and the cooling system is not activated. In this case, in the $m$-th
time slot the system may reach a temperature of $10.1 + (1 + \delta)=11.5$
degrees and the variable $\mathit{stress}$ is $1$.

The process $\mathit{Ctrl}$ will sense a temperature above $10$ sending
the command $\snda {cool} {\on}$ to the actuator $\mathit{cool}$. Now,
since the attack $A_m$ is active in that time slot ($m > 8$), the command
will be stolen by the attack and it will never reach the actuator. Without
that dose of coolant, the temperature of the system will continue to grow.
As a consequence, after further $4$ time units of cooling, i.e.\ in the
$m{+}4$-th time slot, the value of the state variable $\mathit{stress}$
may be $5$ and the system enters in an unsafe state.

After $1$ time slots, in the time slot $m+5$, the controller and the
$\mathit{IDS}$ synchronise via channel $\mathit{sync}$, the $\mathit{IDS}$
will detect a temperature above $10$, and it will fire the output on
channel $\mathit{alarm}$ saying to process $\mathit{Ctrl}$ to keep
cooling. But $\mathit{Ctrl}$ will not send again the command $\snda {cool}
{\on}$. Hence, the temperature would continue to increase and the system
remains in an unsafe state while the process $\mathit{IDS}$ will keep
sending of $\mathit{alarm}$(s) until the invariant of the environment gets
violated.
\end{itemize}
\end{proof}

\begin{proof}[Proof of Proposition~\ref{prop:critical2}]
By induction on the length of the trace.  
\end{proof}

In order to prove \Proposition~\ref{prop:att:integrity}, we introduce Lemma~\ref{lem:sys2}. This is a variant of Lemma~\ref{lem:sys} in which the smart system $\mathit{Sys} $ runs in parallel with the attack $A_n$ defined in Example~\ref{exa:att:integrity}. Here, due to the presence of the attack, the temperature is $2$ degrees higher when compared to the system $\mathit{Sys}$ in isolation. 
\begin{lemma} 
\label{lem:sys2}
Let $\mathit{Sys}$ be the system defined in Example~\ref{exa:sys} and $A_n$ be the attack of Example~\ref{exa:att:integrity}. 
Let
\begin{small}
\begin{displaymath}
\mathit{Sys} \parallel A_n= \mathit{Sys_1}  \trans{t_1}\trans\tick \dots \mathit{Sys_{n-1}}  \trans{t_{n-1}}\trans\tick \mathit{Sys_n} 
\end{displaymath}
\end{small}%
such that the traces $t_j$ contain no $\tick$-actions, for any $j \in  1 .. n{-}1 $,  and for any  $i \in  1 .. n $ $\mathit{Sys_i}= \confCPS {E_i}{P_i} $ with 
$E_i = \envCPS 
{\statefun^i{}} 
{\actuatorfun^i{}} 
{ \delta }  
{\evolmap{}}
{ \epsilon }  
{\measmap{}}   
{\invariantfun{}}$.
Then, for any $i \in 1 .. n{-}1 $ we have the following:
\begin{itemize}[noitemsep]
\item  if   $ \actuatorfun^i{}(\mathit{cool})= \off $, then
 $\statefun^i{}(\mathit{temp})\in [0, 11.1+2+\delta ]$; 

\item if $ \actuatorfun^i{}(\mathit{cool})= \off $ and
$\statefun^i{}(\mathit{temp})\in (10.1+2, 11.1+2+\delta ]$, then we have $
\actuatorfun^{i+1}{}(\mathit{cool}) =\on$;

\item  if  $ \actuatorfun^i{}(\mathit{cool})=\on$, then   $\statefun^i{}(\mathit{temp}) \in ( 9.9+2-k *(1+\delta) , 11.1+2+\delta -k*(1-\delta)] $, 
for some  $k  \in 1..5$,   such that $\actuatorfun^{i-k}{}(\mathit{cool}) =\off $ and $ \actuatorfun^{i-j}{}(\mathit{cool}) =\on $, for $j \in 0..k{-}1$. 

\end{itemize}
\end{lemma}
\begin{proof}
Similar to the proof of    Lemma~\ref{lem:sys}.
\end{proof}

Now, everything is in place to prove Proposition~\ref{prop:att:integrity}. 
\begin{proof}[Proof of Proposition~\ref{prop:att:integrity}]
Let us proceed by case analysis. 
\begin{itemize}[noitemsep]
\item 
Let $0 \leq n \leq 8$.  
In the proof of Proposition~\ref{prop:att:DoS}, we remarked that the system $\mathit{Sys}$ in isolation may sense a temperature greater than $10$  
only after $8$ $\tick$-actions, i.e., in the $9$-th time slot.
However, the life of the attack is $n \leq 8$, and in the $9$-th time
slot the attack is already terminated. As a consequence, starting
from the $9$-th time slot the system will correctly sense the
temperature and it will correctly activate the cooling system.
\item Let $n=9$. 

%% file: macros_tip.tex
\newcommand{\ntrans}[1]{\mathrel{{\trans{#1}}\makebox[0em][r]{$\not$\hspace{2ex}}}{\!}}

\newcommand{\on}{\mathsf{on}}
\newcommand{\off}{\mathsf{off}}

\newcommand{\bool}[1]{\llbracket #1 \rrbracket}

\newcommand{\confCPS}[2]{#1 \, {\Join} \, #2}

\newcommand{\rsens}[2]{\mathsf{read}\, #2(#1)}

%% file: SecLMT-Proof.tex

\section{Proofs}

To prove \autoref{q_and_m_are_metrics} we need some preliminary results.
The first of these results is \autoref{prop_kant_quasimetric} below, which states that the pseudometric property is preserved by function $\Kantorovich$, namely $\Kantorovich(d)$ is a pseudometric over 
$\distrib(\states)$ whenever $d$ is a pseudometric over $\states$.
\autoref{triangolareStrana} supports \autoref{prop_kant_quasimetric}.

\begin{lemma}
\label{triangolareStrana}
Assume two functions $d,d'  \colon \states \times \states \to [0,1]$ with 
$d(t,t') \le d'(t,t'')+d'(t'',t)$ for all terms $t,t',t'' \in \states$.
Then $\Kantorovich(d)(\Delta_1,\Delta_2) \le 
\Kantorovich(d')(\Delta_1,\Delta_3) +
\Kantorovich(d')(\Delta_3,\Delta_2)$ for all distributions $\Delta_1,\Delta_2,\Delta_3 \in \distr{\states}$.
\end{lemma}
\begin{proof}
Consider the function $\omega \colon \states \times \states \to [0,1]$ defined for all terms $t_1,t_2 \in \states$ as 
\[
\omega(t_1,t_2) = \sum_{t_3 \in \states \mid \Delta_3(t_3) \neq 0} \frac{\omega_1(t_1,t_3) \cdot \omega_2(t_3,t_2)}{\Delta_3(t_3)}
\] 
with $\omega_1 \in \Omega(\Delta_1,\Delta_3)$ one of the optimal matchings realising $\Kantorovich(d')(\Delta_1,\Delta_3)$, and $\omega_2 \in \Omega(\Delta_3,\Delta_2)$ one of the optimal matchings realising $\Kantorovich(d')(\Delta_3,\Delta_2)$.
We will prove that:\
\begin{compactenum}
\item \label{Kant_triang_uno}
$\omega$ is a matching in $\Omega(\Delta_1,\Delta_2)$, and
\item \label{Kant_triang_due}
$\sum_{t_1,t_2 \in \states} \omega(t_1,t_2) \cdot d(t_1,t_2) \le \Kantorovich(d')(\Delta_1,\Delta_3)  + \Kantorovich(d')(\Delta_3,\Delta_2)$.
\end{compactenum}
By property \ref{Kant_triang_uno} we infer $\Kantorovich(d)(\Delta_1,\Delta_2) \le \sum_{t_1,t_2 \in \states} \omega(t_1,t_2) \cdot d(t_1,t_2)$, then by property \ref{Kant_triang_due} we infer the thesis
$\Kantorovich(d)(\Delta_1,\Delta_2) \le \Kantorovich(d')(\Delta_1,\Delta_3)  + \Kantorovich(d')(\Delta_3,\Delta_2)$.
To show (\ref{Kant_triang_uno}) we prove that the left marginal of $\omega$ is $\Delta_1$ by\\[3pt]
\begin{math}
\begin{array}{rlr}
& 
\sum_{t_2 \in \states} \omega(t_1,t_2)
\\[5pt]
= \quad & \sum_{t_2 \in \states}  \sum_{t_3 \in \states \mid \Delta_3(t_3) \neq 0} \frac{\omega_1(t_1,t_3) \cdot \omega_2(t_3,t_2)}{\Delta_3(t_3)}
\\[5pt]
= \quad & \sum_{t_3 \in \states \mid \Delta_3(t_3) \neq 0} \frac{\omega_1(t_1,t_3) \cdot \Delta_3(t_3)}{\Delta_3(t_3)} & \text{(by $\omega_2 \in \Omega(\Delta_3,\Delta_2)$)}
\\[5pt]
= \quad & \sum_{t_3 \in \states \mid \Delta_3(t_3) \neq 0} \omega_1(t_1,t_3) 
\\[5pt]
= \quad & \Delta_1(t_1) & \text{(by $\omega_1 \in \Omega(\Delta_1,\Delta_3)$)}
\end{array}
\end{math}\\[5pt]
and we observe that the proof that the right marginal of $\omega$ is $\Delta_2$ is analogous.
Then, we show (\ref{Kant_triang_due}) by\\[3pt]
\begin{math}
\begin{array}{rlr}
&  \sum_{t_1,t_2 \in \states} \omega(t_1,t_2) \cdot d(t_1,t_2)
\\[5pt]
=  \quad &\sum_{t_1,t_2 \in \states} \sum_{t_3 \in \states \mid \Delta_3(t_3) \neq 0} \frac{\omega_1(t_1,t_3) \cdot \omega_2(t_3,t_2)}{\Delta_3(t_3)} \cdot d(t_1,t_2)
\\[5pt]
\le   \quad &  \sum_{t_1,t_2 \in \states,t_3 \in \states \mid \Delta_3(t_3) \neq 0} \frac{\omega_1(t_1,t_3) \cdot \omega_2(t_3,t_2)}{\Delta_3(t_3)} \cdot d'(t_1,t_3)  \; + 
\\[5pt]
 \quad  & 
 \sum_{t_1,t_2 \in \states, t_3 \in \states \mid \Delta_3(t_3) \neq 0} \frac{\omega_1(t_1,t_3) \cdot \omega_2(t_3,t_2)}{\Delta_3(t_3)} \cdot d'(t_3,t_2) 
\\[5 pt]
=  \quad &  \sum_{t_1,t_3 \in \states} \frac{\omega_1(t_1,t_3) \cdot \Delta_3(t_3)}{\Delta_3(t_3)} \cdot  d'(t_1,t_3)  \; + 
   \sum_{t_2,t_3 \in \states} \frac{\Delta_3(t_3) \cdot \omega_2(t_3,t_2)}{\Delta_3(t_3)}  \cdot d'(t_3,t_2)
\\[5pt]
=  \quad &  \sum_{t_1,t_3 \in \states} \omega_1(t_1,t_3) \cdot d'(t_1,t_3) \; +
 \sum_{t_2,t_3 \in \states} \omega_2(t_3,t_2)  \cdot d'(t_3,t_2)
\\[5pt]
=  \quad &  \Kantorovich(d')(\Delta_1,\Delta_3)  + \Kantorovich(d')(\Delta_3,\Delta_2) 
\end{array}
\end{math}\\[5pt]
where the inequality follows from 
the hypothesis and the third last equality follows by $\omega_2 \in \Omega(\Delta_3,\Delta_2)$ and $\omega_1 \in \Omega(\Delta_1,\Delta_2)$.
\qed
\end{proof}

\begin{proposition}
\label{prop_kant_quasimetric}
If $d \colon \states \times \states \to [0,1]$ is a $1$-bounded pseudometric over $\states$, then $\Kantorovich(d) \colon \distrib(\states) \times \distrib(\states) \to [0,1]$ is a $1$-bounded pseudometric over $\distrib(\states)$.
\end{proposition}
\begin{proof} 
We have to prove that $\Kantorovich(d)$ satisfies the three properties in \autoref{def:pseudoquasimetric}.

To show $\Kantorovich(d)(\Delta,\Delta) = 0$ it is enough to take the matching $\omega \in \Omega(\Delta,\Delta)$ defined by $\omega(t,t) = \Delta(t)$, for all $t \in \states$, and $\omega(t,t') = 0$, for all $t,t'\in \states$ with $t \neq t'$. 
In fact, we obtain $\Kantorovich(d)(\Delta,\Delta) =0 $ by $\Kantorovich(d)(\Delta,\Delta) \le
\sum_{t,t' \in \states}\omega(t,t') \cdot d(t,t') = \sum_{t \in \states} \Delta(t) \cdot d(t,t) =0$, with the last equality from the property $d(t,t) =0$ of the pseudometric $d$.

To show the symmetry property $\Kantorovich(d)(\Delta_1,\Delta_2) = \Kantorovich(d)(\Delta_2,\Delta_1)$ it is enough to observe that for any matching $\omega \in \Omega(\Delta_1,\Delta_2)$, the function $\omega' \colon \states \times \states \to [0,1]$ defined for all processes $t_1,t_2 \in \states$ as $\omega'(t_1,t_2) = \omega(t_2,t_1)$, is a matching in $\Omega(\Delta_2,\Delta_1)$. 
In fact, by exploiting this property, given one of the optimal matching $\omega \in \Omega(\Delta_1,\Delta_2)$ realising $\Kantorovich(d)(\Delta_1,\Delta_2)$ we get\\[5pt]
\begin{math}
\begin{array}{rlr}
& \Kantorovich(d)(\Delta_1,\Delta_2)\\[5pt]
= &  \sum_{t_1,t_2 \in \states} \omega(t_1,t_2) \cdot d(t_1,t_2)\\[5pt]
= &  \sum_{t_2,t_1 \in \states} \omega'(t_2,t_1) \cdot d(t_2,t_1)\\[5pt]
\ge & \Kantorovich(d)(\Delta_2,\Delta_1)\\[5pt]
\end{array}
\end{math}\\
with the second equality from the symmetry property $d(t_1,t_2) = d(t_2,t_1)$ of the pseudometric $d$.
Then, by exchanging the role of $\Delta_1$ and $\Delta_2$ we get  
$\Kantorovich(d)(\Delta_2,\Delta_1) \ge \Kantorovich(d)(\Delta_1,\Delta_2)$, thus giving $\Kantorovich(d)(\Delta_1,\Delta_2) = \Kantorovich(d)(\Delta_2,\Delta_1)$.

We conclude by observing that the triangular property 
$\Kantorovich(d)(\Delta_1,\Delta_2) \le \Kantorovich(d)(\Delta_1,\Delta_3) + \Kantorovich(d)(\Delta_3,\Delta_2)$ is an instance of \autoref{triangolareStrana}, which can be applied since the hypothesis $d(t,t') \le d(t,t'') + d(t'',t')$ for all $t,t',t'' \in \states$ follows from the triangular property of the pseudometric $d$.
\qed
\end{proof}


Now we prove that for all $k \ge 1$, the function $\metric^k$ is a fixed point of $\Bisimulation$.
\begin{lemma}
\label{lemma_mk_fixed_point}
For all $k \ge 1$, 
$\Bisimulation(\metric^k) = \metric^k$
\end{lemma}
\begin{proof}
First we note that structure 
${(\{d \colon \states \times \states \to [0,1] \mid \Bisimulation_{\tick}(\metric^{k-1}) \sqsubseteq d\},\sqsubseteq)}$, 
with $d_1 \sqsubseteq d_2$ iff $d_1(t, t') \le d_2(t,t')$ for all $t,t' \in \states$,
is a complete lattice.
Indeed, for each set $D \subseteq [0,1]^{ \states\times  \states}$, the supremum and infimum are defined as $\sup(D)(t,t') = \sup_{d \in D}d(t,t')$ and $\inf(D)(t,t') = \inf_{d \in D}d(t,t')$, for all $t,t' \in \states$. 
The infimum of the lattice is clearly $\Bisimulation_{\tick}(\metric^{k-1})$.
Being $\Bisimulation$ monotone, by the Knaster-Tarski theorem $\Bisimulation$ has a least fixed point.
Since our pLTS is image-finite, and all transitions lead to distributions with finite support, with arguments analogous to those used in~\cite{vB12} it is possible to prove that $\Bisimulation$ 
is continuous 
and its closure ordinal is $\omega$, thus implying that its least fixed point is the supremum of the Kleene ascending chain 
$\Bisimulation_{\tick}(\metric^{k-1})  \sqsubseteq  \Bisimulation(\Bisimulation_{\tick}(\metric^{k-1})) \sqsubseteq  \Bisimulation^2(\Bisimulation_{\tick}(\metric^{k-1}))  \sqsubseteq \ldots$ =
$\metric^{k,0} \sqsubseteq \metric^{k,1} \sqsubseteq \metric^{k,2} \sqsubseteq \ldots$, and, by definition, the supremum of this chain is $\metric^k$.
\end{proof}

Now we exploit \autoref{lemma_mk_fixed_point} to prove that 
for arbitrary processes $t,t' \in \states$, process $t'$ is able to simulate transitions of the form $t \TransStep{\hat \alpha} \Delta$, besides those of the form $t \transStep{\alpha} \Delta$, when $\alpha \neq \tick$.  
\begin{lemma} 
\label{lemma_sim_weak_transitions}
Given two arbitrary terms $t,t' \in \states$,
whenever $t \TransStep{\hat \alpha} \Delta$ for $\alpha \neq \tick$,
we have:
\[
\inf_{t' \TransStep{\hat \alpha} \Theta} \Kantorovich(\metric^k)(\Delta + (1-\size{\Delta}) \dirac{\dummyN},\Theta + (1-\size{\Theta}) \dirac{\dummyN}) \le \metric^k(t,t')
\]
\end{lemma}
\begin{proof}

\noindent
The thesis is immediate if $\metric^k(t,t') =1$.
Consider the case  $\metric^k(t,t') <1$.
We reason by induction on the length $n$ of $t \TransStep{\hat \alpha} \Delta$.

\emph{\underline{Base case $n=1$}}.
In this case $t \TransStep{\hat \alpha} \Delta$ is directly derived from $t \transStep{\hat \alpha} \Delta$.
There are two sub-cases.
The first is $\alpha = \tau$ and $\Delta = \dirac{t}$, the second is $t \transSim{\alpha} \Delta$, with $\alpha$ an arbitrary action in $\Act \setminus \{\tick\}$.  
In the former case, by the definition of the weak transition relation $\trans{\widehat \tau}$ we have that $t' \trans{\widehat \tau} \dirac{t'}$ and, consequently, $t' \TransStep{\widehat \tau} \dirac{t'}$. 
The thesis holds for distribution $\Theta = \dirac{t'}$.
More precisely, we have that $\Kantorovich(\metric^k)(\dirac{t} + (1-\size{\dirac{t}})\dirac{\dummyN},\dirac{t'} + (1-\size{\dirac{t'}})\dirac{\dummyN}) = \Kantorovich(\metric^k)(\dirac{t},\dirac{t'}) = \metric^k(t,t')$.
In the latter case,
the thesis follows directly by \autoref{def:metric_sim_functional} and 
\autoref{lemma_mk_fixed_point}.
In detail, \autoref{def:metric_sim_functional} gives
\[
\inf_{t' \TransStep{\hat \alpha} \Theta} \Kantorovich(\metric^k)(\Delta,\Theta + (1-\size{\Theta}) \dirac{\dummyN}) \le \Bisimulation(\metric^k)(t,t')
\]
and \autoref{lemma_mk_fixed_point} gives $\Bisimulation(\metric^k)(t,t') = \metric^k(t,t')$.

\emph{\underline{Inductive step $n>1$}}.
The derivation $t \TransStep{\hat \alpha} \Delta$ is obtained by $t \TransStep{\hat \beta_1} \Delta'$ and $\Delta' \transStep{\hat{\beta}_2} \Delta$, for some distribution $\Delta' \in \distr{\states}$ and actions $\beta_1,\beta_2 \in \Act \setminus \{\tick\}$.
We have two sub-cases. 
The first is $\beta_1=\tau$ and $\beta_2=\alpha$, the other is $\beta_1=\alpha$ and $\beta_2=\tau$.
We consider the case $\beta_1=\tau$ and $\beta_2=\alpha$, the other is analogous.

The length of derivation $t \TransStep{\hat \beta_1} \Delta'$ is $n-1$. 
Therefore, by the inductive hypothesis we have
\begin{equation}
\label{lemma_sim_weak_transitions_ind_hyp}
\inf_{t' \TransStep{\hat{\beta}_1} \Theta'} \Kantorovich(\metric^k)(\Delta' + (1-\size{\Delta'}) \dirac{\dummyN},\Theta' + (1-\size{\Theta'}) \dirac{\dummyN}) \le \metric^k(t,t')
\end{equation}
Notice that $\metric^k(t,t') < 1$ and \autoref{lemma_sim_weak_transitions_ind_hyp} ensure that the set $\{\Theta' \mid t' \TransStep{\hat{\beta}_1} \Theta'\}$ is not be empty.
Moreover, being $\beta_1 = \tau$, we have that $\size{\Delta'} = 1$ and, for each transition $t' \TransStep{\hat{\beta}_1} \Theta'$, also $\size{\Theta'} = 1$.
Therefore,
the inductive hypothesis \autoref{lemma_sim_weak_transitions_ind_hyp}
instantiates to 
\begin{equation}
\label{lemma_sim_weak_transitions_ind_hyp_instance}
\inf_{t' \TransStep{\hat{\beta}_1} \Theta'} \Kantorovich(\metric^k)(\Delta' ,\Theta' ) \le \metric^k(t,t')
\end{equation}

The sub-distribution $\Delta'$ is of the form $\Delta' = \sum_{i \in I}p_i \cdot \dirac{t_i}$ for suitable processes $t'_i$ and, by definition of transition relation $\transSim{\hat{\beta}_2}$, 
the transition $\Delta' \transSim{\hat{\beta}_2} \Delta$ is derived from a $\beta_2$-transition by some of the processes $t_i$, namely $I$ is partitioned into sets $I_1 \cup I_2$ such that:\
\begin{inparaenum}[(i)] 
\item 
for all $i \in I_1$ we have $t_i \transStep{\beta_2} \Delta_i$ for suitable distributions $\Delta_i$, 
\item
for each $i \in I_2$ we have $t_i \ntransStep{\beta_2}$, and 
\item $\Delta = \sum_{i \in I_1} p_i \cdot \Delta_i$.
\end{inparaenum}

Let us fix an arbitrary transition $t' \TransStep{\hat \beta_1} \Theta'$ (remember we argued above that it is not possible that there are none). 
The sub-distribution
$\Theta'$ is of the form $\Theta' = \sum_{j \in J}q_j \cdot \dirac{t'_j}$, for suitable processes $t'_j$.
Then, $J$ can be partitioned into sets $J_1 \cup J_2$ such that for all $j \in J_1$ we have $t'_j \TransStep{\hat{\beta}_2} \Theta_j$ for suitable distributions $\Theta_j$ and for each $j \in J_2$ we have $t'_j\nTransStep{\hat{\beta}_2}$. 
If $J_1 \neq \emptyset$
this gives $\Theta' \TransStep{\hat{\beta}_2} \Theta$ with $\Theta = \sum_{j \in J_1} q_j \cdot \Theta_j$.
Since we had $t' \TransStep{\hat{\beta}_1} \Theta'$, we can conclude $t' \TransStep{\hat{\alpha}} \Theta$.
Notice that we are sure that there exist some some $\Theta'$ with $t' \TransStep{\hat \beta_1} \Theta'$ for which $J_1 \neq \emptyset$.
Indeed, if for all  $\Theta'$ with $t' \TransStep{\hat \beta_1} \Theta'$ we had $J_1 = \emptyset$, this would cause $t \nTransStep{\hat{\alpha}}$, giving $\Bisimulation(\metric^k)(t,t') =1$ and contradicting $\Bisimulation(\metric^k)(t,t')  = \metric^k(t,t') < 1$.
We remark that in all cases where $J_1 \neq \emptyset$, 
the weak transition
$t' \TransStep{\hat{\alpha}} \Theta$ is obtained by firstly choosing one of the available weak transitions labelled $\hat{\beta}_1$ from $t'$, namely $t' \TransStep{\hat{\beta}_1} \Theta'$, and, then, by choosing one of the available weak transitions labelled $\beta_2$ from $t'_j$, namely $t'_j \TransStep{\hat{\beta}_2} \Theta_j$, for all $j \in J_1$.
%

For the transition
$t' \TransStep{\hat{\beta}_1} \Theta'$ fixed above,
let $\omega$ be one of the optimal matchings realising $\Kantorovich(\metric^k)(\Delta' ,\Theta')$.
We can rewrite the distributions $\Delta'$ and $\Theta'$ as 
$\Delta' = \sum_{i \in I, j \in J} \omega(t_i,t'_j) \cdot \dirac{t_i}$ and 
$\Theta' = \sum_{i \in I, j \in J} \omega(t_i,t'_j)  \cdot \dirac{t'_j}$.
For all $i \in I_1$ and $j \in J$, define $\Delta_{i,j} = \Delta_i$.
We can rewrite $\Delta$ as $\Delta = \sum_{i \in I_1,j\in J} \omega(t_i,t'_j) \cdot \Delta_{i,j}$.
Analogously, for each $j \in J_1$ and $i \in I$ we note that the transition $q_j \dirac{t'_j} \TransStep{\hat{\beta}_2} q_j \cdot  \Theta_{j}$ can always be split into $\sum_{i \in I} \omega(t_i,t'_j)\dirac{t'_j} \TransStep{\hat{\beta}_2} \sum_{i \in I} \omega(t_i,t'_j)  \cdot \Theta_{i,j}$
so that
we can rewrite $\Theta_j$ as $\Theta_j = \sum_{i \in I}\omega(t_i,t'_j)  \cdot \Theta_{i,j}$ and $\Theta$ as
$\Theta = \sum_{i\in I, j \in J_1} \omega(t_i,t'_j)  \cdot  \Theta_{i,j}$. 
Then we note that for all $i \in I_1$ and $j \in J_1$, 
all transition $t'_j \TransStep{\hat{\beta}_2} \Theta_{i,j}$ 
ensure that
\begin{equation}
\label{lemma_sim_weak_transitions_interne}
\inf_{t'_j \TransStep{\hat{\beta}_2} \Theta_{i,j}} \Kantorovich(\metric^k)(\Delta_{i,j},\Theta_{i,j} + (1-\size{\Theta_{i,j}})\dirac{\dummyN}) \le  \metric^k(t_i,t'_j)
\end{equation}
Indeed, 
by definition of $\Bisimulation$,
whenever $t_i \transStep{\beta_2} \Delta_i = \Delta_{i,j}$ we have
\[
\inf_{t'_j \TransStep{\hat{\beta}_2} \Theta_{i,j}} \Kantorovich(\metric^k)(\Delta_{i,j},\Theta_{i,j} + (1-\size{\Theta_{i,j}})\dirac{\dummyN}) \le \Bisimulation(\metric^k)(t_i,t'_j) 
\]
Then, being $\metric^k$ a fixed point of $\Bisimulation$ we have $\Bisimulation(\metric^k)(t_i,t'_j)  = \metric^k(t_i,t'_j)$ and \autoref{lemma_sim_weak_transitions_interne} follows.

Consider any 
$j \in J_1$ and $i \in I_1$.
By \autoref{lemma_sim_weak_transitions_interne} and $\Bisimulation(\metric^k)(t_i,t'_j)  = \metric^k(t_i,t'_j)$, we infer that if 
if $\metric^k(t_i,t'_j) <1$, then the set of the weak transitions labelled $\hat{\beta}_2$ from $t'_j$ cannot be empty.
For any transition $t'_i \TransStep{\hat{\beta}_2} \Theta_{i,j}$,
let $\omega_{i,j}$ be one of the optimal matchings realising $\Kantorovich(\metric^k)(\Delta_{i,j}, \Theta_{i,j} + (1-\size{\Theta_{i,j}})\dirac{ \dummyN})$.
Define $\omega' \colon \states \times \states \to [0,1]$ as the function such that for arbitrary processes $u,v \in \states$ we have:
\[ 
\omega'(u,v) = 
\begin{cases} 
\displaystyle
\sum_{i \in I_1, j \in J_1}\omega(t_i,t'_j)  \omega_{i,j}(u,v) 
& \text{ if } u \neq \dummyN \neq v\\[4pt]
\displaystyle\sum_{i \in I_1, j \in J_1}\omega(t_i,t'_j) \omega_{i,j}(u,v) + 
 \sum_{i \in I_1, j \in J_2}\omega(t_i,t'_j) \Delta_{i,j}(u)
& \text{ if } u \neq \dummyN = v \\[4pt]
\displaystyle \sum_{i \in I_1, j \in J_1}\omega(t_i,t'_j)  \omega_{i,j}(u,v) + 
 \sum_{i \in I_2, j \in J_1}\omega(t_i,t'_j) \Theta_{i,j}(v)
& \text{ if } u = \dummyN \neq v \\[4pt]
\displaystyle\sum_{i \in I_1, j \in J_1}\omega(t_i,t'_j)  \omega_{i,j}(u,v) + 
 \sum_{i \in I_1, j \in J_2}\omega(t_i,t_j) \Delta_{i,j}(u) + \\
\displaystyle \sum_{i \in I_2, j \in J_1}\omega(t_i,t'_j)  \Theta_{i,j}(v)+
\sum_{i \in I_2, j \in J_2}\omega(t_i,t'_j) 
& \text{ if } u = \dummyN = v. 
\end{cases}
\] 

To infer the proof obligation 
\begin{equation}
\label{lemma_sim_weak_transitions_pobb}
\inf_{t' \TransStep{\hat \alpha} \Theta} \Kantorovich(\metric^k)(\Delta + (1-\size{\Delta}) \dirac{\dummyN},\Theta + (1-\size{\Theta}) \dirac{\dummyN}) \le \metric^k(t,t')
\end{equation}
it is now enough to show that:
\begin{enumerate}
\item 
\label{matching} 
the function $\omega'$ is a matching in $\Omega(\Delta + (1-\size{\Delta}) \dirac{\dummyN} ,\Theta + (1-\size{\Theta}) \dirac{\dummyN})$
\item 
\label{metric_condition} 
\begin{equation}
\label{target}
\inf_{t' \TransStep{\hat{\beta}_1} \Theta' \atop{
\Theta' = \sum_{j \in J_1 \cup J_2} q_j \delta(t'_j) \atop
t'_j \TransStep{\hat{\beta}_2} \Theta_{i,j}}}  
\sum_{u,v \in \states} \omega'(u,v) \cdot \metric^k(u,v) \le \metric^k(t,t')
\end{equation}
\end{enumerate}

To show property \ref{matching} we prove that the left marginal of $\omega'$ is $\Delta + (1-\size{\Delta}) \dirac{\dummyN}$.
The proof that the right marginal is $\Theta + (1-\size{\Theta}) \dirac{\dummyN}$ is analogous.
For any process $u \neq \dummyN$ we have\\[5pt]
\begin{math}
\begin{array}{rlr}
& \sum_{v \in \states}\omega'(u,v) 
\\[5pt]
=
&
\sum_{v \neq \dummyN} \sum_{i \in I_1, j \in J_1}\omega(t_i,t'_j) \omega_{i,j}(u,v)
\\[5pt] 
&
+
\sum_{i \in I_1, j \in J_1}\omega(t_i,t'_j) \omega_{i,j}(u,\dummyN)
+ 
\sum_{i \in I_1, j \in J_2}\omega(t_i,t'_j) \Delta_{i,j}(u)
\\[5pt]
= 
&
\sum_{i \in I_1, j \in J_1}\omega(t_i,t'_j) \sum_{v \in \states} \omega_{i,j}(u,v)
+ 
\sum_{i \in I_1, j \in J_2}\omega(t_i,t'_j) \Delta_{i,j}(u)  
\\[5pt]
=
&
\sum_{i \in I_1, j \in J_1}\omega(t_i,t'_j) \Delta_{i,j}(u)
+ 
\sum_{i \in I_1, j \in J_2}\omega(t_i,t'_j) \Delta_{i,j}(u)
\\[5pt]
=
&
\sum_{i \in I_1, j \in J}\omega(t_i,t'_j)  \Delta_{i,j}(u)
\\[5pt]
=
&
\sum_{i \in I_1}p_i \Delta_{i}(u)
\\[5pt]
=
&
\Delta(u)
\\[5pt]
=
& (\Delta + (1-\size{\Delta})\dirac{ \dummyN})(u)
\\[5pt]
\end{array}
\end{math}

\noindent
with the third equality from the fact that  $\omega_{i,j}$ is a matching in $\Omega(\Delta_{i,j},\Theta_{i,j})$, the fourth equality by $J = J_1 \cup J_2$ and the fifth equality by
$\sum_{j \in J}\omega(t_i,t'_j) =  p_i$ and $\Delta_{i,j} = \Delta_i$.
Consider now $\dummyN$. We have
\\[5pt]
\begin{math}
\begin{array}{rlr}
&
\sum_{v \in \states}\omega'(\dummyN,v)
\\[5pt]
= 
&
\sum_{v \neq \dummyN}\sum_{i \in I_1, j \in J_1}\omega(t_i,t'_j) \omega_{i,j}(\dummyN,v)
+ \sum_{v \neq \dummyN} \sum_{i \in I_2, j \in J_1}\omega(t_i,t'_j) \Theta_{i,j}(v) 
\\[5pt]
&
+ 
\sum_{i \in I_1, j \in J_1}\omega(t_i,t'_j) \omega_{i,j}(\dummyN,\dummyN)
+
\sum_{i \in I_1, j \in J_2}\omega(t_i,t'_j) \Delta_{i,j}(\dummyN)
\\[5pt]
&
+ \sum_{i \in I_2, j \in J_1}\omega(t_i,t'_j)  \Theta_{i,j}(\dummyN)
+ \sum_{i \in I_2, j \in J_2}\omega(t_i,t'_j)
\\[5pt]
=
&
\sum_{v \in \states}\sum_{i \in I_1, j \in J_1}\omega(t_i,t'_j) \omega_{i,j}(\dummyN,v) 
+
\sum_{v \in \states} \sum_{i \in I_2, j \in J_1}\omega(t_i,t'_j)  \Theta_{i,j}(v) 
\\[5pt]
&
+
\sum_{i \in I_1, j \in J_2}\omega(t_i,t'_j) \Delta_{i,j}(\dummyN) +
\sum_{i \in I_2, j \in J_2}\omega(t_i,t'_j)
\\[5pt]
=
&
\sum_{i \in I_1, j \in J_1}\omega(t_i,t'_j) \Delta_{i,j}(\dummyN)
+
\sum_{i \in I_2, j \in J_1}\omega(t_i,t'_j)
\\[5pt]
&
+
\sum_{i \in I_1, j \in J_2}\omega(t_i,t'_j) \Delta_{i,j}(\dummyN) +
\sum_{i \in I_2, j \in J_2}\omega(t_i,t'_j)
\\[5pt]
=
&
\sum_{i \in I_1, j \in J}\omega(t_i,t'_j)  \Delta_{i,j}(\dummyN) + 
\sum_{i \in I_2, j \in J}\omega(t_i,t'_j) 
\\[5pt]
=
&
\sum_{i \in I_1} p_i \Delta_{i}(\dummyN) + 
\sum_{i \in I_2}p_i
\\[5pt]
=
&
(\Delta + (1-\size{\Delta}) \dirac{\dummyN})(\dummyN)
\\[5pt]
\end{array}
\end{math}

\noindent
where the third equality by the fact that 
$\omega_{i,j}$ is a matching in $\Omega(\Delta_{i,j},\Theta_{i,j})$ and the fact that
$\Theta_{i,j}$ is a distribution, the fourth equality by $J = J_1 \cup J_2$, the fifth equality by  $\sum_{j \in J}\omega(t_i,t'_j) =  p_i$ and $\Delta_{i,j} = \Delta_i$ and
the last equality follows from  $\sum_{i \in I_1, j \in J}\omega(s_i,t_j) = \sum_{i \in I_1} p_i = \size{\Delta}$.

Summarising, for all processes $u \in \states$ we have proved that $\sum_{v \in \states}\omega'(u,v) = (\Delta + (1-\size{\Delta}) \dirac{\dummyN})(u)$, thus confirming that the left marginal of $\omega'$ is $\Delta + (1-\size{\Delta}) \dirac{\dummyN}$.

To prove (\ref{metric_condition}), by looking at the definition of $\omega'$ given above we get that $\sum_{u,v \in \states} \omega'(u,v) \cdot \metric^k(u,v)$ is the summation of the following values:
\begin{itemize}
\item
$\sum_{u \neq \dummyN \neq v} \sum_{i \in I_1, j \in J_1}\omega(t_i,t'_j)  \omega_{i,j}(u,v) \metric^k(u,v)$\\
\item
$\sum_{u \neq \dummyN}\sum_{i \in I_1, j \in J_1}\omega(t_i,t'_j)  \omega_{i,j}(u,\dummyN)  \metric^k(u,\dummyN)$ + \\
$\sum_{i \in I_1, j \in J_2}\omega(t_i,t'_j) \Delta_{i,j}(u) \metric^k(u,\dummyN)$\\
\item
$\sum_{v \neq \dummyN}\sum_{i \in I_1, j \in J_1}\omega(t_i,t'_j)   \omega_{i,j}(\dummyN,v)  \metric^k(\dummyN,v)$ + \\$
\sum_{i \in I_2, j \in J_1}\omega(t_i,t'_j)  \Theta_{i,j}(v)  \metric^k(\dummyN,v)$\\
\item
$\sum_{i \in I_1, j \in J_1}\omega(t_i,t'_j) \omega_{i,j}(\dummyN,\dummyN)  \metric^k(\dummyN,\dummyN)$
+ \\
$ \sum_{i \in I_1, j \in J_2}\omega(t_i,t'_j)  \Delta_{i,j}(\dummyN) \metric^k(\dummyN,\dummyN)$
+\\
 $\sum_{i \in I_2, j \in J_1}\omega(t_i,t'_j)  \Theta_{i,j}(\dummyN) \metric^k(\dummyN,\dummyN)$+\\
 $\sum_{i \in I_2, j \in J_2}\omega(t_i,t'_j) \metric^k(\dummyN,\dummyN)$.
\end{itemize}

By moving the first summand of the second, third and fourth items to the first item, we rewrite this summation as the summation of the following values:
\begin{itemize}
\item
$\sum_{u, v \in \states} \sum_{i \in I_1, j \in J_1} \omega(t_i,t'_j) \omega_{i,j}(u,v)  \metric^k(u,v)$\\
\item
$\sum_{i \in I_1, j \in J_2}\omega(t_i,t'_j)  \Delta_{i,j}(u) \metric^k(u,\dummyN)$\\
\item
$\sum_{i \in I_2, j \in J_1}\omega(t_i,t'_j) \Theta_{i,j}(v) \metric^k(\dummyN,v)$\\
\item
$\sum_{i \in I_1, j \in J_2}\omega(t_i,t'_j) \Delta_{i,j}(\dummyN)  \metric^k(\dummyN,\dummyN)
+  \\ 
\sum_{i \in I_2, j \in J_1}\omega(t_i,t'_j)  \Theta_{i,j}(v)\metric^k(\dummyN,\dummyN) + \\
 \sum_{i \in I_2, j \in J_2}\omega(t_i,t'_j) \metric^k(\dummyN,\dummyN)$.
\end{itemize}

Since the function
$\omega_{i,j}$ was defined as one of the optimal matchings realising $\Kantorovich(\metric^k)(\Delta_{i,j}, \Theta_{i,j} + (1-\size{\Theta_{i,j}})\dirac{ \dummyN})$,
the first item can be rewritten as $\sum_{i \in I_1, j \in J_1} \omega(t_i,t'_j) \Kantorovich(\metric^k)(\Delta_{i,j},\Theta_{i,j}+ (1-\size{\Theta_{i,j}})\dirac{ \dummyN})$.
From \autoref{lemma_sim_weak_transitions_interne} we get
$\inf_{t'_j \TransStep{\hat{\beta}_2} \Theta_{i,j}}\Kantorovich(\metric^k)(\Delta_{i,j},\Theta_{i,j} + (1-\size{\Theta_{i,j}})\dirac{ \dummyN})) \le \metric^k(t_i,t'_j)$.
Henceforth the infimum for all $t'_j \TransStep{\hat{\beta}_2} \Theta_{i,j}$ of 
the first item is less or equal $\sum_{i \in I_1, j \in J_1} \omega(t_i,t'_j) \cdot  \metric^k(t_i,t'_j)$.
The second item is clearly less or equal than $\sum_{i \in I_1, j \in J_2}\omega(t_i,t'_j)$.
The third item is clearly less or equal than $\sum_{i \in I_2, j \in J_1}\omega(t_i,t'_j)$.
Finally, the last item is 0 since $\metric^k(\dummyN,\dummyN) = 0$.
Namely, 
the infimum for all $t'_j \TransStep{\hat{\beta}_2} \Theta_{i,j}$ of
$\sum_{u,v \in \states} \omega'(u,v) \cdot  \metric^k(u,v)$ is bounded by the summation of the following three values:
\begin{itemize}
\item
$\sum_{i \in I_1, j \in J_1} \omega(t_i,t'_j)   \metric^k(t_i,t'_j)$ 
\item 
$\sum_{i \in I_1, j \in J_2}\omega(t_i,t'_j)$ 
\item
$\sum_{i \in I_2, j \in J_1}\omega(t_i,t'_j)$.
\end{itemize}
Formally:
\begin{equation}
\label{primoInf}
\begin{split}
\inf_{t'_j \TransStep{\hat{\beta}_2} \Theta_{i,j}}
\sum_{u,v \in \states} \omega'(u,v) \cdot  \metric^k(u,v)  \quad \quad \quad \quad \quad \quad \quad \quad  \quad \quad \quad\quad \quad \quad \quad \quad \quad \\
\le 
\sum_{i \in I_1, j \in J_1} \omega(t_i,t'_j)   \metric^k(t_i,t'_j) +
\sum_{i \in I_1, j \in J_2}\omega(t_i,t'_j)+
\sum_{i \in I_2, j \in J_1}\omega(t_i,t'_j)
\end{split}
\end{equation}
Then, since $\Kantorovich(\metric^k)(\Delta' ,\Theta' )$ is the summation of the following values:
\begin{itemize}
\item
$\sum_{i \in I_1,j\in J_1} \omega(t_i,t'_j)   \metric^k(t_i,t'_j)$ 
\item
$\sum_{i \in I_1,j\in J_2} \omega(t_i,t'_j)  \metric^k(t_i,t'_j) = \sum_{i \in I_1,j\in J_2} \omega(t_i,t'_j)$ (since $t_i \transStep{\beta_2}$ and $t'_j \nTransStep{\hat{\beta_2}}$ give $ \metric^k(t_i,t'_j) =1$) 
\item
$\sum_{i \in I_2,j\in J_1} \omega(t_i,t'_j)  \metric^k(t_i,t'_j)= \sum_{i \in I_2,j\in J_1} \omega(t_i,t'_j)$ (since $t'_j \transStep{\beta_2}$ and $t_i \nTransStep{\hat{\beta_2}}$ give $ \metric^k(t_i,t'_j) =1$) 
\item
$\sum_{i \in I_2,j\in J_2} \omega(t_i,t'_j)  \metric^k(t_i,t'_j)$
\end{itemize}
we infer that the right hand side of \autoref{primoInf} $\sum_{i \in I_1, j \in J_1} \omega(t_i,t'_j) \cdot  \metric^k(t_i,t'_j) + \sum_{i \in I_1, j \in J_2}\omega(t_i,t'_j) + \sum_{i \in I_2, j \in J_1}\omega(t_i,t'_j)$ is less or equal than $\Kantorovich(\metric^k)(\Delta' ,\Theta' )$.
Together with \autoref{primoInf} this gives
\begin{equation}
\inf_{t'_j \TransStep{\hat{\beta}_2} \Theta_{i,j}}
\sum_{u,v \in \states} \omega'(u,v) \cdot  \metric^k(u,v) \le 
\Kantorovich(\metric^k)(\Delta' ,\Theta' )
\end{equation}
which, together with \autoref{lemma_sim_weak_transitions_ind_hyp_instance} gives \autoref{target}, which concludes the proof.
\qed
\end{proof}

We are now ready to prove Theorem~\ref{q_and_m_are_metrics}.

\begin{proof}[of  \autoref{q_and_m_are_metrics}]
We have to prove that $\metric^k$ satisfies the three properties in \autoref{def:pseudoquasimetric}.
Properties $\metric^k(t,t) = 0$ and $\metric^k(t,t') = \metric^k(t',t)$ for all $t,t' \in \states$ are immediate.
%
The interesting case is  the triangular property 
$\metric^k(t,t') \le \metric^k(t,t'') + \metric^k(t'',t')$ for all $t,t',t'' \in \states$.
To this purpose, let us define the function $\metric \colon \states \times \states \to [0,1]$ such that
\[
\metric(t,t') = \min\Big(\metric^{k}(t,t') , \inf_{t'' \in \states} (\metric^k(t,t'') + \metric^k(t'',t'))\Big).
\]

We will prove that $\metric = \metric^k$.
By the definition of $\metric$, this gives $\metric^k(t,t') \le \metric^k(t,t'') + \metric^k(t'',t')$ for all $t'' \in \states$, thus confirming that also the triangular property holds for $\metric^k$.

In order to prove $\metric = \metric^k$, we observe first that relation $\metric \sqsubseteq \metric^k$ follows immediately by the definition of $\metric$.
It remains to prove $\metric^k \sqsubseteq \metric$.
To this purpose we prove that: 
\begin{inparaenum}[(i)]
\item 
\label{qm_metric_po1}
$\metric^k$ is the least prefixed point of 
the functional $\Bisimulation$ on the complete lattice 
${(\{d \colon \states \times \states \to [0,1] \colon \Bisimulation_{\tick}(\metric^{k-1}) \sqsubseteq d\},\sqsubseteq)}$, and
\item
\label{qm_metric_po2}
$\metric$ is a prefixed point of the same functional on the same lattice.
\end{inparaenum}

Let us start with property (\ref{qm_metric_po1})
By \autoref{lemma_mk_fixed_point}, $\metric^k$ is the least fixed point of the functional $\Bisimulation$, which is monotone and continuous in the lattice.
This coincides with the least prefixed point.

Let us consider now (\ref{qm_metric_po2}).
We have to prove $\Bisimulation(\metric) \sqsubseteq \metric$, namely,
whenever $\metric(t,t') < 1$, then, for all $\alpha \neq \tick$ we have
\begin{equation}\label{op}
\forall t \trans{\alpha} \Delta . 
\inf_{t' \TransStep{\hat{\alpha}} \Theta}   \Kantorovich(\metric)(\Delta,\Theta + (1-\size{\Theta}) \dirac{\dummyN}) \le \metric(t,t').
\end{equation}
To prove \autoref{op} we distinguish two cases, namely $\metric(t,t') = \metric^k(t,t')$ and 
$\metric(t,t') = \inf_{t'' \in \states} (\metric^k(t,t'') + \metric^k(t'',t'))$.

Assume first $\metric(t,t') = \metric^k(t,t')$.
In this case, 
being $\metric^k$ the least fixed point of $\Bisimulation$,
$t \trans{\alpha} \Delta$ implies that 
\[
\inf_{t' \TransStep{\hat{\alpha}} \Theta} \Kantorovich(\metric^k)(\Delta,\Theta + (1-\size{\Theta}) \dummyN) \le \Bisimulation(\metric^k)(t,t') = \metric^k(t,t') = \metric(t,t')
\]
Since $\Kantorovich$ is monotone and $\metric \sqsubseteq \metric^k$, we infer 
\[
\inf_{t' \TransStep{\hat{\alpha}} \Theta} \Kantorovich(\metric)(\Delta,\Theta + (1-\size{\Theta}) \dummyN) \le \metric(t,t')
\] 
thus giving \autoref{op}.

Assume now $\metric(t,t') = \inf_{t'' \in \states} (\metric^k(t,t'') + \metric^k(t'',t'))$.
Since $\metric(t,t') < 1$, there exist terms $t'' \in \states$ with $\metric^k(t,t'') + \metric^k(t'',t') < 1$, thus implying both $\metric^k(t,t'')<1$ and  $\metric^k(t'',t') < 1$.
By \autoref{lemma_sim_weak_transitions}, 
from $\metric^k(t,t'') < 1$ and
$t \trans{\alpha} \Delta$ we infer
\[
\inf_{t'' \TransStep{\hat{\alpha}} \Phi}
\Kantorovich(\metric^k)(\Delta,\Phi + (1-\size{\Phi}) \dirac{\dummyN}) \le \metric^k(t,t'')
\]
By \autoref{lemma_sim_weak_transitions}, from $\metric^k(t'',t') < 1$, for all $t'' \Trans{\hat{\alpha}} \Phi$ 
we have 
\[
\inf_{t' \TransStep{\hat{\alpha}} \Theta} \Kantorovich(\metric^k)(\Phi + (1-\size{\Phi})\dirac{\dummyN},\Theta + (1-\size{\Theta}) \dirac{\dummyN}) \le \metric^k(t'',t')
\]
By the definition of $\metric$ and \autoref{triangolareStrana} we have 
$\Kantorovich(\metric^k)(\Delta,\Phi + (1-\size{\Phi}) \dummyN) + 
\Kantorovich(\metric^k)(\Phi + (1-\size{\Phi})\dummyN),\Theta + (1-\size{\Theta}) \dummyN)
\ge 
\Kantorovich(\metric)(\Delta,\Theta + (1-\size{\Theta}) \dummyN)$.
We derive 
\[
\inf_{t'' \TransStep{\hat{\alpha}} \Phi \atop t' \TransStep{\hat{\alpha}} \Theta}
 \Kantorovich(\metric)(\Delta,\Theta + (1-\size{\Theta}) \dummyN) 
\le
\metric^k(t,t'') + \metric^k(t'',t')
\]
and, by definition of infimum,
\[
\inf_{t'' \TransStep{\hat{\alpha}} \Phi \atop t' \TransStep{\hat{\alpha}} \Theta}
 \Kantorovich(\metric)(\Delta,\Theta + (1-\size{\Theta}) \dummyN) 
\le
\metric(t,t') 
\]
which gives \autoref{op} and concludes the proof. 
\qed
%
\end{proof}

We prove now \autoref{prop_simulazione}.

\begin{proof}[of \autoref{prop_simulazione}]
We prove the first item, then the second item follows by the first and the result $t \simeq_0 t'$ iff $t \approx t'$ given in \cite{DJGP02}.
First we recall that $t \simeq_p t'$ iff $\metric(t,t') = p$, where $\metric$ is the least fixed point (and also least prefixed point) in the lattice $([0,1]^{\states \times \states},\sqsubseteq)$ of a functional $\Bisimulation'$ such that $\Bisimulation'(d)(t,t') = \max(\Bisimulation(d)(t,t') ,\Bisimulation_{\tick}(d)(t,t'))$ for all $t,t' \in \states$ and $d \in  [0,1]^{\states \times \states}$.
Therefore, we have to prove that $\metric^{\infty} = \metric$.

Let us start with $\metric^{\infty} \sqsubseteq \metric$.
Being $\metric^{\infty}$ the supremum of all $\metric^k$,
it is enough to show $\metric^k \sqsubseteq \metric$ for all $k \in \bbbn$.
This property can be shown by induction over $k$.
The base case is immediate since $\metric^0 = \zeroF$.
Consider the inductive step $k+1$.
Function $\metric^{k+1}$ is obtained as
$\sup_{n \to \infty } \Bisimulation^n(\Bisimulation_{\tick}(\metric^k))$. 
Assume any $n \in \bbbn$. 
By 
$\Bisimulation' \ge \Bisimulation,\Bisimulation_{\tick}$ we get
$\Bisimulation^n(\Bisimulation_{\tick}(\metric^k)) \sqsubseteq
(\Bisimulation')^{n+1}(\metric^k)$ for all $n \in \bbbn$.
By the monotonicity of $\Bisimulation'$ and the inductive hypothesis $\metric^k \sqsubseteq \metric$, we get $(\Bisimulation')^{n+1}(\metric^k) \sqsubseteq (\Bisimulation')^{n+1}(\metric)$.
Finally, since $\metric$ is a fixed point of $\Bisimulation'$ we infer $ (\Bisimulation')^{n+1}(\metric) = \metric$.
Summarising, $\Bisimulation^n(\Bisimulation_{\tick}(\metric^k)) \sqsubseteq \metric$.
By the arbitrarity of $n$ we infer $\metric^{\infty} \sqsubseteq \metric$.

Let us show now  $\metric \sqsubseteq \metric^{\infty}$.
Being $\metric$ the least prefixed point of $\Bisimulation'$, it is enough to show that
$\metric^{\infty}$ is a prefixed point of $\Bisimulation'$.
We have both  $\metric^{\infty} \sqsupseteq \Bisimulation(\metric^{\infty})$ and
$\metric^{\infty} \sqsupseteq \Bisimulation_{\tick}(\metric^{\infty})$, thus giving
$\metric^{\infty} \sqsupseteq \Bisimulation'(\metric^{\infty})$, confirming that $\metric^{\infty}$ is a prefixed point of $\Bisimulation'$.
\qed
\end{proof}


Now we prove  \autoref{thm:attack-tolerance-gen}.
\begin{proof}[of \autoref{thm:attack-tolerance-gen}]
We prove the second item.
The proof of the third item is analogous, then the first item is a consequence of the others.
To prove the thesis we can prove that for all $k \in \bbbn$ we have
\begin{align*}
& \metric^{k}(\confCPS \xi  P_1 \parallel P_2 \parallel A, \confCPS \xi  P_1 \parallel P_2) \\
\le \quad & \metric^{k}(\confCPS \xi  P_1  \parallel A,\confCPS \xi  P_1) +
\metric^{k}(\confCPS \xi  P_2  \parallel A,\confCPS \xi  P_2 ) \\
- \quad & (\metric^{k}(\confCPS \xi  P_1  \parallel A,\confCPS \xi  P_1  ) \cdot \metric^{k}(\confCPS \xi  P_2  \parallel A,\confCPS \xi  P_2  ))
\end{align*}
Since 
$\confCPS \xi  P_1 \parallel P_2 \parallel A$ can mimic all the behaviours by
$\confCPS \xi  P_1 \parallel P_2$, the distance 
$\metric^{k}(\confCPS \xi  P_1 \parallel P_2 \parallel A, \confCPS \xi  P_1 \parallel P_2)$ is given by the behaviours by
$\confCPS \xi  P_1 \parallel P_2 \parallel A$ that are not mimicked by
$\confCPS \xi  P_1 \parallel P_2$.
Then, since 
$\confCPS \xi  P_1 \parallel A \parallel P_2 \parallel A$ can mimic all the behaviours by
$\confCPS \xi  P_1 \parallel P_2 \parallel A$, we have that  
\[\metric^{k}(\confCPS \xi  P_1 \parallel P_2 \parallel A, \confCPS \xi  P_1 \parallel P_2) 
\le
\metric^{k}(\confCPS \xi  P_1 \parallel A \parallel P_2 \parallel A, \confCPS \xi  P_1 \parallel P_2) 
\]
thus implying that to have the proof obligation we can prove the stronger property
\begin{align*}& \metric^{k}(\confCPS \xi  P_1 \parallel A \parallel P_2 \parallel A, \confCPS \xi  P_1 \parallel P_2) \\
\le \quad & \metric^{k}(\confCPS \xi  P_1  \parallel A,\confCPS \xi  P_1) +
\metric^{k}(\confCPS \xi  P_2  \parallel A,\confCPS \xi  P_2 ) \\
- \quad & (\metric^{k}(\confCPS \xi  P_1  \parallel A,\confCPS \xi  P_1  ) \cdot \metric^{k}(\confCPS \xi  P_2  \parallel A,\confCPS \xi  P_2  )).
\end{align*}
More in general, we prove
\begin{align*}
& \metric^{k}(\confCPS \xi  Q_1 \parallel Q_2, \confCPS \xi  P_1 \parallel P_2) \\
\le \quad & \metric^{k}(\confCPS \xi  Q_1  ,\confCPS \xi  P_1) +
\metric^{k}(\confCPS \xi  Q_2,\confCPS \xi  P_2 ) \\
- \quad & (\metric^{k}(\confCPS \xi  Q_1,\confCPS \xi  P_1  ) \cdot \metric^{k}(\confCPS \xi  Q_2 ,\confCPS \xi  P_2  ))\\
\end{align*}
for arbitrary $Q_1$ and $Q_2$,
written also
\begin{equation}\label{pob}
\begin{split}
 \metric^{k}(M_1 \parallel M_2, N_1 \parallel N_2) \le \qquad \qquad \qquad \qquad \qquad \qquad \qquad \qquad  \qquad \\
\metric^{k}(M_1  ,N_1) +
\metric^{k}(M_2 ,N_2 ) 
-  (\metric^{k}(M_1,N_1 ) \cdot \metric^{k}(M_2 ,N_2)).
\end{split}
\end{equation}

To this purpose, first we need to introduce the notion of congruence closure for $\metric^{k}$ as the quantitative analogue of the well-known concept of congruence closure of a process equivalence. 
We define the metric congruence closure of $\metric^{k}$ for operator $\parallel$ w.r.t.\ the bound provided in \autoref{pob}
as a function $m$ assigning to each pair of systems a distance in $[0,1]$ given by
\[
	m(M,N)=
	\begin{cases}
		 \min(m(M_1,N_1) + m(M_2,N_2) - (m(M_1,N_1)  m(M_2,N_2)), \metric^{k}(M,N)) \\
                    \qquad \qquad \qquad \text{if } 
			\left[\begin{array}{l}
				M=M_1\parallel M_{2} \land
				N=N_1 \parallel N_{2} \,\land \\
				\metric^{k}(M_1,N_1) < 1 \,\land 
				\metric^{k}(M_2,N_2) < 1
			\end{array} \right.\\
		\metric^{k}(M,N) \quad \text{otherwise}
	\end{cases}
\]

We note that $m$ satisfies by construction 
$m(M_1 \parallel M_1,N_1 \parallel N_2) \le m(M_1,N_1) + m(M_2,N_2) - (m(M_1,N_1)  \cdot m(M_2,N_2))$.
We note also that $m$ satisfies by construction $m \sqsubseteq \metric^{k}$.
It remains to show that $\metric^{k} \sqsubseteq m$, thus giving $\metric^{k} = m$, and \autoref{pob} holds. 
Since $\metric^{k}$ is the least prefixed point of $\Bisimulation$ over the lattice
${(\{d \colon \states \times \states \to [0,1] \mid \Bisimulation_{\tick}(\metric^{k-1}) \sqsubseteq d\},\sqsubseteq)}$
to show $\metric^{k} \sqsubseteq m$ it is enough to prove that $m$ is a prefixed point of the same functional on the same lattice.

To prove that 
$\Bisimulation(m) \sqsubseteq m$ we need to show that $m$ satisfies the transfer condition of the bisimulation metrics, namely
\begin{equation}\label{bound_parallel_proof_obligation}
 \forall M \trans{\alpha} \gamma . \; \exists M' \Trans{\alpha} \gamma' . \; \Kantorovich(m)(\gamma , \gamma' + (1-|\gamma'|)\overline{\dummyN}) \le m(M,M')
\end{equation}
for all systems $M,M$ with $m(M,M') < 1$ and $\alpha \neq \tick$.

This can be proved by applying the same arguments used to prove Proposition 3.2 in \cite{GLT16}.
\qed
\end{proof}

\noindent
\textbf{Proof of \autoref{prop:case1}} \hspace{0.2 cm}
First we observe that in the evolution of both systems $\confCPS {\xi }  {\mathit{Ctrl_i}$ and $\confCPS {\xi } { \mathit{Ctrl_i}\parallel   A_{\mathsf{fp}} \langle i , m, n \rangle}}$ it never happens that there are more than two instantaneous actions in between any two $\tick$ actions.
This implies that for all $j \in \bbbn$, system $M$ reachable from $\confCPS {\xi }  {\mathit{Ctrl_i}$ and system $N$ reachable from $\confCPS {\xi } { \mathit{Ctrl_i}\parallel   A_{\mathsf{fp}} \langle i , m, n \rangle}}$, we have $\metric^{j}(M,N) =  \sup_{h \in \bbbn} \metric^{j,h}(M,N)= \metric^{j,2}(M,N)$.
Then, the proof follows from the following 7 properties, by observing that first item of the thesis follows from the property expressed by \autoref{prop:case-prop0} below and the second and third items of the thesis follow from the property expressed by \autoref{prop:case-prop3c} below,
 when, respectively, $j_1=j-m+1$ and $j_2=m-1$.
For any $j \in \bbbn$, it holds that:
\begin{enumerate}
\item  
\label{prop:case-prop0}
$\metric^{j,l}(\confCPS{\xi}{P} ,\confCPS{\xi}{P \parallel Q}) =0$ 
for any  $P$  and whenever  process $Q$ has the form
$Q=\tick^{j'} . B\langle i,  n-m+1  \rangle$ for some $j < j'$.

%
\item  
\label{prop:case-prop1a}
$\metric^{j,0}(\confCPS{\xi}{P} ,\confCPS{\xi}{P   \parallel Q})  =
  1-(p_i^+)^{j-1}$ 
whenever $0<  j  \leq n-m  +1$,  
$\xi (r_i)=\mathsf{absence}$,
and the processes $P$ and $Q$ have the form 
$P =\tick .\mathit{Ctrl_{i }}$    
and  
$Q=B\langle i,  n-m+1-j \rangle$. 

\item  
\label{prop:case-prop1b}
$\metric^{j,1}(\confCPS{\xi}{P} ,\confCPS{\xi}{P  \parallel Q}) =
1-(p_i^+)^{j-1}$ 
whenever $0<  j  \leq n-m  +1$, $\xi (r_i)=\mathsf{absence}$, 
and the processes $P$ and $Q$ have the form  
$P  = \snda{c_i}{\mathsf{on}} .\tick.\mathit{Ctrl_{i}}$
and   
$Q=B\langle i,  0,n-m+1-j   \rangle$. 

\item  
\label{prop:case-prop1c}
$\metric^{j}(\confCPS{\xi}{P} ,\confCPS{\xi}{P \parallel Q}) =
1-(p_i^+)^{j}$ 
whenever       
$0<j  \le  n-m   +1$, $\xi (r_i)=\mathsf{absence}$, 
and the processes $P$ and $Q$ have the form  
$P = \mathit{Ctrl_{i}}$  
and
$Q=  B\langle i,  n-m+1 \rangle$. 

%
%
%
%
%

\item  
\label{prop:case-prop3a}
$\metric^{j,0}(\confCPS{\xi}{P} ,\confCPS{\xi}{P \parallel Q}) =
1-(p_i^+)^{j_1}$ 
whenever   
processes $P$ and $Q$ have the form 
$P = \tick. \mathit{Ctrl_i}$ and
$Q=\tick^{j_2} . B\langle i,  n-m+1 \rangle$, for some
 $0 < j_2 \le j$  such that $j_1=\min(j-j_2+1,    n-m+1)$.

\item  
\label{prop:case-prop3b}
$\metric^{j,1}(\confCPS{\xi}{P} ,\confCPS{\xi}{P \parallel Q}) =
 1-(p_i^+)^{j_1}$ 
whenever    
processes $P$ has either the form
$P =  \snda{c_i}{\mathsf{on}} .\tick.\mathit{Ctrl_i}$ or
 $P =  \snda{c_i}{\mathsf{off}} .\tick. \mathit{Ctrl_i}$, and process
$Q$ has the form 
$Q=\tick^{j_2} . B\langle i,  n-m+1 \rangle$,
for some $0 < j_2 \le j$ such that  $j_1=\min(j-j_2+1,    n-m+1 )$.

\item  
\label{prop:case-prop3c}
$\metric^{j}(\confCPS{\xi}{P} ,\confCPS{\xi}{P \parallel Q}) =
1-(p_i^+)^{j_1}$ 
whenever  
processes $P$ and $Q$ have the form form $P = \mathit{Ctrl_i}$
and
  $Q=\tick^{j_2} . B\langle i,  n-m+1 \rangle$, for some
$0 < j_2 \le j $  such that $j_1=\min(j-j_2+1,    n-m+1)$.

%
%

%
%

\end{enumerate}
 
The seven properties above can be proved for all $\metric^{j}$ and $\metric^{j,l}$
by well founded induction over the relation $\prec$ defined as  follows:
\begin{itemize}
\item $\metric^{j} \prec  \metric$ if  $\metric \in \{\metric^{j' },\metric^{j',l}\}$ with $j<j'$
\item $\metric^{j,l}\prec \metric$ if either  $\metric \in \{\metric^{j' },\metric^{j',l}\}$ with $j<j'$,  or,
$\metric =\metric^{j',l'} $ with $j'=j$ and $l<l'$.
\end{itemize}
Obviously, $\prec$ is irreflexive and there does not exist any infinite  descending chain (the base case is  $\metric^{0}$).

The base case $j=0$ is immediate since $\metric^{0}$ is the constant zero function $\zeroF$ and $1-(p_i^+)^{0}=0$.

We consider the inductive step.
\begin{enumerate}
\item 

The thesis can be easily proved since  $Q$ can perform 
only $\tick$ actions and, intuitively, it does not affect the behaviour of $P$.

In detail, for $j=1$ and $l=0$, we have that  whenever $\confCPS{\xi}{P}\trans{\tick} \sum_{i \in I} \confCPS{\dirac{\xi_i}}{\dirac{P_i}} $, 
then $ \confCPS{\xi}{P\parallel Q }\trans{\tick} \sum_{i \in I} \confCPS{\dirac{\xi_i}}{\dirac{P_i\parallel Q'}} $ with $Q=   \tick^{j'-1} .B\langle i,  n-m+1  \rangle$. 
Hence the thesis follows, since $\metric^{0}(\confCPS{\dirac{\xi_i}}{\dirac{P_i}} ,
\confCPS{\dirac{\xi_i}}{\dirac{P_i\parallel Q'}}  )=0$ by definition of $\metric^{0}$.

Assume now 
$l>0$. 
In this case, whenever $\confCPS{\xi}{P}\trans{\alpha} \sum_{i \in I} \confCPS{\dirac{\xi }}{\dirac{P_i}}$  with $\alpha\neq \tick$, 
then $ \confCPS{\xi}{P\parallel Q }\trans{\alpha} \sum_{i \in I} \confCPS{\dirac{\xi}}{\dirac{P_i \parallel Q}}$. 
The thesis holds since, by induction on 
case \autoref{prop:case-prop0}, we have 
$\metric^{j,l-1}(\confCPS{\dirac{\xi}}{\dirac{P_i}}, \confCPS{\dirac{\xi}}{\dirac{P_i\parallel Q}}  )=0$.

Similarly, for $l=0$ and $j>1$, whenever $\confCPS{\xi}{P}\trans{\tick} \sum_{i \in I} \confCPS{\dirac{\xi_i}}{\dirac{P_i}} $, 
then $ \confCPS{\xi}{P\parallel Q }\trans{\tick} \sum_{i \in I} \confCPS{\dirac{\xi_i}}{\dirac{P_i\parallel Q'}}$ with $Q'=\tick^{j'-1} . B\langle i,  n-m+1  \rangle$. 
Hence the thesis holds since, by induction on  
case \autoref{prop:case-prop0}, for any $h$, it holds that  
$ \metric^{j-1,h}(\confCPS{\dirac{\xi_i}}{\dirac{P_i}}, \confCPS{\dirac{\xi_i}}{\dirac{P \parallel Q'}} )=0$
thus implying that  
\[
 \metric^{j-1 }(\confCPS{\dirac{\xi_i}}{\dirac{P_i}}, \confCPS{\dirac{\xi_i}}{\dirac{P \parallel Q'}} )= 
\sup_{h \in \bbbn^{\infty}} \metric^{j-1,h}(\confCPS{\dirac{\xi_i}}{\dirac{P_i}}, \confCPS{\dirac{\xi_i}}{\dirac{P \parallel Q'}} )=0
.
\]

%
%
%
\item 
Define $M=\confCPS{\xi}{P}$ and $N=\confCPS{\xi}{P\parallel Q }$.
We have that $\metric^{j.0}(M ,N) = \Bisimulation_{\mathit{\tick}}( \metric^{j-1})(M,N)= \Bisimulation_{\mathit{\tick}}( \metric^{j-1,2})(M,N)$.
Hence we have to prove that $\Bisimulation_{\mathit{\tick}}( \metric^{j-1,2})(M,N)=1-(p_i^+)^{j-1 }$.
Such a property follows by the following two facts:
\begin{compactitem}
\item
$ \displaystyle \max_{M \transStep{\tick}\Delta} \min_{N \TransStep{\tick} \Theta} \Kantorovich(\metric^{j-1,2})(\Delta,\Theta + (1-\size{\Theta})\dirac{\dummyN}) =1-(p_i^+)^{j -1}$\\[0.5 ex]
\item
$ \displaystyle \max_{N \transStep{\tick}\Theta} \min_{M \TransStep{\tick} \Delta} \Kantorovich(\metric^{j-1,2})(\Delta + (1-\size{\Delta}) \dirac{\dummyN}, \Theta ) =  1-(p_i^+)^{j -1 }$.\\[0.5 ex]
\end{compactitem}

We prove with the first case, the second one is similar.

The only transitions by $M$ are of the form 
$M  \trans{\tick}  \confCPS{\dirac{\xi'}}{\dirac{\mathit{Ctrl_{i}} }}$
with 
  $\xi' \in \mathit{next}(\xi)$. 
The environments $\xi' \in \mathit{next}(\xi)$ maximising the set 
\[
  \min_{N \TransStep{\tick} \Theta} \Kantorovich(\metric^{j-1,2})( \confCPS{\dirac{\xi'}}{\dirac{\mathit{Ctrl_{i}} }},
  \Theta ) 
 \]
 are  such that $\statefun' (r_i)=\mathsf{absence}$.
Indeed the attacker could force
 $N$ to perform  $ \snda {c_i}  {\mathsf{on}} $ with probability equal to $1$.
If $\xi ' (r_i)=\mathsf{absence}$, then $M$ will perform   $ \snda {c_i} {\mathsf{on}} $ 
with probability $  p_i^+ $. Hence  $M$ does not simulate $N$ with a probability $1-p_i^+$.
Otherwise, if $\xi ' (r_i)=\mathsf{presence}$, then $M$ will perform   $ \snda {c_i}  {\mathsf{on}} $ 
with probability $1- p_i^- $. Hence  $M$ does not simulate $N$ with a probability $ p_i^-$.
Since $0 \leq p_i^+, p_i^- < \frac{1}{2}$, then $1- p_i^+ > p_i^-$.


The system $N=  \confCPS{\xi}{P\parallel Q } $ minimises 
\[
  \min_{N \TransStep{\tick} \Theta} \Kantorovich(\metric^{j-1,2})( \confCPS{\dirac{\xi'}}{\dirac{\mathit{Ctrl_{i}} }},
  \Theta ) 
 \]
 by simulating  $M$ with   the transition $N \trans{\tick} \dirac{\confCPS{\xi'}{\mathit{Ctrl_i}\parallel Q'}} $ 
 with    $Q'=B\langle i, \max(0,n-m+1-j-1)\rangle$.

The only admissible matching $\omega$  for  
$ \Kantorovich(\metric^{j-1,2} )( \dirac{\confCPS{\xi'}{\mathit{Ctrl_i}}}
  \dirac{\confCPS{\xi'}{\mathit{Ctrl_i}\parallel Q'}}) $
  is such that 
$\omega( \dirac{\confCPS{\xi'}{\mathit{Ctrl_{i}}}} ,  \dirac{\confCPS{\xi'}{\mathit{Ctrl_i}\parallel Q'}})= 1$.

Summarising we have:
 \[
\begin{array}{rlr}
& \max_{M \transStep{\tick}\Delta} \min_{N \TransStep{\tick} \Theta} \Kantorovich( \metric^{j-1,2})(\Delta,\Theta + (1-\size{\Theta})\dirac{\dummyN}) \\[2.0 ex]
= &   \min_{N \TransStep{\tick} \Theta} \Kantorovich( \metric^{j-1,2})( \dirac{\confCPS{\xi'}{\mathit{Ctrl_i}}} ,  \Theta ) 
 \q\q\q\q\q\q  \text{ with $\xi' (r_i)=\mathsf{absence}$}
\\[2.0 ex]
=&  \Kantorovich(\metric^{j-1,2} )( \dirac{\confCPS{\xi'}{\mathit{Ctrl_i}}} ,  
\dirac{\confCPS{\xi'}{\mathit{Ctrl_i}\parallel Q'}})  \\[2.0 ex]
= &  \metric^{j-1,2} 
( \dirac{\confCPS{\xi'}{\mathit{Ctrl_i} }} ,
  \dirac{\confCPS{\xi'}{\mathit{Ctrl_i}\parallel Q'}})    
  \q\q\q\q\q\q    \text{(by induct. on case \autoref{prop:case-prop1c})} 
 \\[2.0 ex]
= & 1-(p_i^+)^{j-1 } 
\end{array}
\] 
which completes the  the proof.

\item 
Define  $M=\confCPS{\xi}{P }$ and $N=\confCPS{\xi}{P\parallel Q }$.

Analogously to
\autoref{prop:case-prop1a}, to prove $\Bisimulation ( \metric^{j,0})(M,N)=1-(p_i^+)^{j -1}$ 
it is sufficient to prove  the following two facts: 
\begin{compactitem}
\item
$  \max_{M \transStep{\snda {c_i} {\mathsf{on}}}\Delta} \min_{N \TransStep{\snda {c_i} {\mathsf{on}}} \Theta} \Kantorovich(\metric^{j ,0})(\Delta,\Theta + (1-\size{\Theta})\dirac{\dummyN}) =1-( p_i^+)^{j -1}$
\item
$   \max_{N \transStep{\snda {c_i} {\mathsf{on}}}\Theta} \min_{M \TransStep{\snda {c_i} {\mathsf{on}}} \Delta} \Kantorovich(\metric^{j ,0})(\Delta + (1-\size{\Delta}) \dirac{\dummyN}, \Theta ) = 1-( p_i^+)^{j -1}$. 
\end{compactitem}

We prove  the first case, the second one is similar.

The only transition by $M=\confCPS{\xi}{P }$    is 
$M \trans{\snda {c_i} {\mathsf{on}}} \confCPS{\dirac{\xi }}{\dirac{\tick.\mathit{Ctrl_{i}} }}$.
The only transition by  $N=  \confCPS{\xi}{P\parallel Q } $ is 
 $N \trans{\snda {c_i} {\mathsf{on}}}   \confCPS{\dirac{\xi }}{\dirac{\tick.\mathit{Ctrl_{i}} \parallel Q}}  $.

The only admissible matching $\omega$  for  
$ \Kantorovich(\metric^{j-1,0} )(   \confCPS{\dirac{\xi }}{\dirac{\tick.\mathit{Ctrl_{i}} }},
  \confCPS{\dirac{\xi }}{\dirac{\tick.\mathit{Ctrl_{i}}\parallel Q } } ) $
  is such that 
$\omega(   \confCPS{\dirac{\xi }}{\dirac{\tick.\mathit{Ctrl_{i}} }},
  \confCPS{\dirac{\xi }}{\dirac{\tick.\mathit{Ctrl_{i}} \parallel Q} } ) = 1$.

Summarising we have:
 \[
\begin{array}{rlr}
& \max_{M \transStep{\snda {c_i} {\mathsf{on}}}\Delta} \min_{N \TransStep{\snda {c_i} {\mathsf{on}}} \Theta} \Kantorovich( \metric^{j-1,0})(\Delta,\Theta + (1-\size{\Theta})\dirac{\dummyN}) \\[2.0 ex]
= &   \min_{N \TransStep{\snda {c_i} {\mathsf{on}}} \Theta} 
\Kantorovich( \metric^{j-1,0})( \dirac{\confCPS{\xi }{\tick.\mathit{Ctrl_i}}} ,  \Theta ) 
 \\[2.0 ex]
=&  \Kantorovich(\metric^{j-1,0} )( \dirac{\confCPS{\xi }{\tick.\mathit{ Ctrl_i}}} ,  
\dirac{\confCPS{\xi }{\tick.\mathit{Ctrl_i\parallel Q} }})  \\[2.0 ex]
= &  \metric^{j-1,0} 
( \dirac{\confCPS{\xi }{\tick.\mathit{Ctrl_i} }} ,
  \dirac{\confCPS{\xi }{\tick.\mathit{Ctrl_i\parallel Q}}})    
  \q\q\q\q\q\q    \text{(by induct. on case \autoref{prop:case-prop1a})} 
 \\[2.0 ex]
= & 1-(p_i^+)^{j-1 } 
\end{array}
\] 
which completes the  the proof.

\item 
Define $M=\confCPS{\xi}{P }$ and $N=\confCPS{\xi}{P\parallel Q }$.

Since $\metric^{j }= \metric^{j ,2}$, analogously to
\autoref{prop:case-prop1a}, to prove $\Bisimulation ( \metric^{j,1})(M,N)=1-(p_i^+)^{j }$
it is sufficient to prove  the following two facts: 
\begin{compactitem}
\item
$ \max_{M \transStep{\tau}\Delta} \min_{N \TransStep{\tau} \Theta} \Kantorovich(\metric^{j ,1})(\Delta,\Theta + (1-\size{\Theta})\dirac{\dummyN}) \le 1-( p_i^+)^{j }$
\item
$  \max_{N \transStep{\tau}\Theta} \min_{M \TransStep{\tau} \Delta} \Kantorovich(\metric^{j ,1})(\Delta + (1-\size{\Delta}) \dirac{\dummyN}, \Theta ) = 1-( p_i^+)^{j }$.
\end{compactitem}

The interesting case is the second.
Indeed, $N$ is always able to simulate $M$ by considering the case in which the controller reads the right value of the sensor and  does not take the value provided by the attacker.
The system $N=  \confCPS{\xi}{P\parallel Q }$ can perform two transitions depending on the fact that the controller reads or not the fake value provided by the attacker.
But,  obviously, 
the system $N=  \confCPS{\xi}{P\parallel Q } $ maximises 
\[ 
 \max_{N \transStep{\tau}\Theta} \min_{M \TransStep{\tau} \Delta} \Kantorovich(\metric^{j ,1})(\Delta + (1-\size{\Delta}) \dirac{\dummyN}, \Theta )  
 \]
when the controller reads the fake value, namely by the transition 
$
N \TransStep{\widehat \tau} \gamma_N=  \dirac{N'  } 
$
where $N'= \confCPS{\xi}{\snda{c_i}{\mathsf{on}} . \tick.\mathit{Ctrl_{i }}}  $.

The  system $M=\confCPS{\xi}{P }$ minimises 
\[ 
 \min_{M \TransStep{\tau} \Delta} \Kantorovich(\metric^{j ,1})(\Delta + (1-\size{\Delta}) \dirac{\dummyN}, \gamma_N)
 \]
by simulating $N$ by the transition  
\[
M \trans{\tau} \gamma_M=  (p_i^+) \cdot \dirac{M_1}+ (1-p_i^+) \cdot \dirac{M_2}
\]
where $M_1= \confCPS{\xi}{ \snda{c_i}{\mathsf{on}} . \tick.\mathit{Ctrl_{i }}} $ and    
 $M_2= \confCPS{\xi}{\snda{c_i}{\mathsf{off}} . \tick.\mathit{Ctrl_{i }}} $.

Moreover, the only admissible   matching  $\omega$  for 
$  \Kantorovich(\metric^{j.1} )(\gamma_M,  \gamma_N) $
  is such that
$\omega(M_1,N')= p_i^+ $ and $\omega(M_2,N')= 1-p_i^+$.

Summarising:
\[
\begin{array}{rlr}
&    \max_{N \transStep{\tau}\Theta} \min_{M \TransStep{\tau} \Delta} \Kantorovich(\metric^{j ,1})(\Delta + (1-\size{\Delta}) \dirac{\dummyN}, \Theta )  \\[2.0 ex]
= &  \min_{M \TransStep{\tau} \Delta} \Kantorovich(\metric^{j ,1})(\Delta + (1-\size{\Delta}) \dirac{\dummyN}, \gamma_N)\\[2.0 ex]
=&  \Kantorovich(\metric^{j.1} )(\gamma_M,  \gamma_N)     \\[2.0 ex]
= &  (p_i^+) \cdot \metric^{j ,1}(M_1,N')+ (1-p_i^+) \cdot \metric^{j ,1}(M_2,N') \\[2.0 ex]
= &  (p_i^+) \cdot (1-( p_i^+)^{j-1} )  + (1-p_i^+) \cdot 1 
 \q\q\q\q\q\q    \text{(by induct. on case \autoref{prop:case-prop1b})}\\[2.0 ex]
= & 1-(p_i^+)^{j } .
\end{array}
\] 
which completes the proof.

\item  The proof is similar to  the proof of  \autoref{prop:case-prop1a}. Indeed  this
 case can be proved by induction on case \autoref{prop:case-prop1c} if $j=j_1$ and $j_2=1 $, and, 
on case \autoref{prop:case-prop3c} if $j_2>1 $.

\item  The proof is similar to  the proof of  \autoref{prop:case-prop1b}. Indeed  this case can be proved by induction on case \autoref{prop:case-prop3a}. 

\item   The proof is similar to  the proof of  \autoref{prop:case-prop1c}. Indeed  this case can be proved by induction on case \autoref{prop:case-prop3b}.

\end{enumerate}
\qed

\noindent
\textbf{Proof of \autoref{prop:case2}} \hspace{0.2 cm}
 The proof is similar to that of \autoref{prop:case1} by considering  $p_i^-$ instead of $p_i^+$, and, 
$C\langle  \ldots  \rangle$ instead of $B\langle  \ldots  \rangle$. 
\qed